\documentclass[a4paper,USenglish,cleveref,autoref,thm-restate]{lipics-v2021}

  \newif\iflong
\newif\ifshort

\longtrue

\iflong
\else
 \shorttrue
\fi

\iflong
\hideLIPIcs
\nolinenumbers
\fi

 \usepackage{amsmath,stmaryrd,graphicx}
 \newcommand{\directions}{\{\uparrow,\downarrow,\leftarrow,\rightarrow\}}
  \newcommand{\bidirections}{\{\updownarrow,\leftrightarrow\}}
\newcommand{\pred}{\textsf{pred}}

\makeatletter
\newcommand{\fixed@sra}{$\vrule height 2\fontdimen22\textfont2 width 0pt\shortrightarrow$}
\newcommand{\shortarrow}[1]{   \mathrel{\text{\rotatebox[origin=c]{\numexpr#1*45}{\fixed@sra}}}
}

\usepackage{verbatim}
\usepackage[textsize=footnotesize]{todonotes}
\usepackage{xcolor}
\usepackage[textsize=footnotesize]{todonotes}

\DeclareMathOperator{\bdist}{\textsf{bd}}
\DeclareMathOperator{\bvect}{\textsf{vect}}

\usepackage{tcolorbox}
\usepackage{cite}

\usepackage{etex}
\usepackage{enumerate}
\usepackage{graphics}
\usepackage{xspace}

\usepackage{caption}
\usepackage{subcaption}
\usepackage{epsfig}
\usepackage{arcs}
\usepackage[T1]{fontenc}
\usepackage{alltt}
\usepackage{amssymb,amsfonts,epsf,url}
\usepackage{graphicx}
\usepackage{pstricks,pstricks-add}
\usepackage{pst-node}
\usepackage{pst-coil}
\usepackage{amsmath}
\usepackage{amsthm}
\usepackage{verbatim}
\usepackage[textsize=footnotesize]{todonotes}
\usepackage{xcolor}
\usepackage[textsize=footnotesize]{todonotes}
\usepackage{fancyhdr}
\usepackage{tikz}
\usetikzlibrary{patterns}

\usetikzlibrary{arrows}

\newcommand{\cmpl}{\textsc{\textup{GCMP}}\xspace} 
\newcommand{\gcmp}{\textsc{\textup{CMP}}\xspace} 
\newcommand{\dcmp}{\textsc{\textup{CMP-D}}\xspace}

\theoremstyle{plain}
 
\RequirePackage{etex}

\newcommand{\bigoh}{\mathcal{O}}

\newtheorem{assumption}[theorem]{Assumption}
 
\newtheorem{reduction}{Reduction Rule}[section]

\theoremstyle{definition}

\newlength{\alginputwidth}

\newcommand{\YES}{\textup{\textsf{YES}}}

\newcommand{\Oh}{{\mathcal O}}
\newcommand{\nat}{\mathbb{N}}

\newcommand{\F}{\mathcal{F}}
\newcommand{\M}{\mathcal{M}}

\newcommand{\NP}{\mbox{{\sf NP}}}
\newcommand{\tw}{\mbox{{\sf tw}}}

\newcommand{\FPT}{\mbox{{\sf FPT}}}

\newcommand{\R}{\mathcal{R}}
\newcommand{\SSS}{\mathcal{S}}
\newcommand{\CCC}{\mathcal{C}}
\newcommand{\dist}{\emph{\textnormal{dist}}}
\newcommand{\sect}{\texttt{\textup{Sect}}}

\newcommand{\T}{\mathcal{T}}

\newcommand{\I}{\mathcal{I}}

\newenvironment{tightcenter}
 {\parskip=0pt\par\nopagebreak\centering}
 {\par\noindent\ignorespacesafterend}

\usepackage{tikz}
\usetikzlibrary{calc}
\usepackage{xargs}
\usepackage{xifthen}
\usepackage{framed}

\usepackage{ctable}

\newlength{\RoundedBoxWidth}
\newsavebox{\GrayRoundedBox}
\newenvironment{GrayBox}[1]%
   {\setlength{\RoundedBoxWidth}{\textwidth-10.5ex}
    \def\boxheading{#1}
    \begin{lrbox}{\GrayRoundedBox}
       \begin{minipage}{\RoundedBoxWidth}%
   }{%
       \end{minipage}
    \end{lrbox}%
    \begin{tightcenter}%
    \begin{tikzpicture}%
       \node(Text)[draw=black!90,fill=white,rounded corners,%
             inner sep=2ex,text width=\RoundedBoxWidth]%
             {\usebox{\GrayRoundedBox}};
        \coordinate(x) at (current bounding box.north west);
        \node [draw=white,rectangle,inner sep=3pt,anchor=north west,fill=white] 
        at ($(x)+(6pt,.75em)$) {\boxheading};
    \end{tikzpicture}
    \end{tightcenter}\vspace{0pt}%
    \ignorespacesafterend
}    

\newenvironment{problem}[2][]{\noindent\ignorespaces%
                                \FrameSep=8pt%
                                \parindent=0pt%
                \vspace*{-.5em}
                \ifthenelse{\isempty{#1}}{%
                  \begin{GrayBox}{\textsc{#2}}%
                }{%
                }
                \newcommand\Prob{{Problem:}}%
                \newcommand\Input{{Input:}}%
                          
                \begin{tabular*}{\textwidth}{@{\hspace{.1em}} >{\itshape} p{1.2cm} p{0.85\textwidth} @{}}%
            }{
                \end{tabular*}%
                \end{GrayBox}%
                \vspace*{-.5em}
                \ignorespacesafterend
            } 
 
       \title{Coordinated Motion Planning is \FPT{} on Discretized Simple Polygons}

 \author{Argyrios Deligkas}{Department of Computer Science, Royal Holloway, University of London, Egham, UK}{Argyrios.Deligkas@rhul.ac.uk}{https://orcid.org/0000-0002-6513-6748}{Supported by Engineering and Physical Sciences Research Council (EPSRC) grant EP/X039862/1.}

 \author{Eduard Eiben}{Department of Computer Science, Royal Holloway, University of London, Egham, UK}{eduard.eiben@rhul.ac.uk}{https://orcid.org/0000-0003-2628-3435}{ Supported by Engineering and Physical Sciences Research Council (EPSRC) grant UKRI4530: Exploring Parameterized Complexity in Blockchain Systems.}

 \author{Robert Ganian}{Algorithms and Complexity Group, TU Wien, Vienna, Austria}{rganian@gmail.com}{https://orcid.org/0000-0002-7762-8045}{Project No. Y1329 of the Austrian Science Fund (FWF), Project No. ICT22-029 of the Vienna Science Foundation (WWTF).}

 \author{Iyad Kanj}{School of Computing, DePaul University, Chicago, USA}{ikanj@cdm.depaul.edu}{0000-0003-1698-8829}{DePaul URC Grants 606601 and 350130.}

\authorrunning{Argyrios Deligkas, Eduard Eiben, Robert Ganian, Iyad Kanj}  
 \Copyright{Argyrios Deligkas, Eduard Eiben, Robert Ganian, Iyad Kanj}  

\category{}  
\iflong
\relatedversion{A short version of this manuscript appears in the proceedings of ICALP 2026.}   
\fi
\ifshort
\relatedversion{Full version: *****arxiv*****}   
\fi

\acknowledgements{}

 \EventEditors{Sayan Bhattacharya, Danupon Nanongkai, Michael Benedikt, and Gabriele Puppis}
\EventNoEds{4}
\EventLongTitle{53rd International Colloquium on Automata, Languages, and Programming (ICALP 2026)}
\EventShortTitle{ICALP 2026}
\EventAcronym{ICALP}
\EventYear{2026}
\EventDate{July 7--10, 2026}
\EventLocation{Royal Holloway, University of London, Egham, United Kingdom}
\EventLogo{}
\SeriesVolume{374}
\ArticleNo{42}

\begin{document}

\maketitle
 
\begin{abstract}
In the coordinated motion planning problem, we are given a graph together with the starting and destination vertices of 
$k$ robots. At each time step, any subset of robots may move, each traversing an edge of the graph, provided that no two robots collide. 
The goal is to compute a schedule that routes all robots to their destinations while minimizing some objective function. In this paper, we focus on the well-studied objective of minimizing the total travel length of all robots. This problem is known to be 
\NP-hard, and it has been shown to be fixed-parameter tractable (\FPT), when parameterized by the number $k$
of robots, on full grids (SoCG 2023) and on bounded-treewidth graphs (ICALP 2024). 

We present a fixed-parameter algorithm for coordinated motion planning, parameterized by the number $k$ of robots, on graphs arising from discretizations of simple polygons.
Such graphs are of particular interest in real-world applications, where planar motion is often constrained to discretized representations of polygonal environments. Moreover, these graphs generalize rectangular grids; consequently, our result constitutes a significant step toward resolving the parameterized complexity of coordinated motion planning on subgrids and, ultimately, planar graphs---two prominent open problems in the field.
 \end{abstract}
  
  \newpage

\section{Introduction}
 In the \textsc{Coordinated Motion Planning} problem (CMP)---also known as \textsc{Multi-Agent Pathfinding}---we are given a graph $G$ and a set $\mathcal{R}$ of robots, each with a designated start vertex and destination vertex in $G$. 
 The goal is to compute a \emph{schedule} that routes all robots to their destinations while minimizing a given objective. Such a schedule consists of a sequence of discrete time steps in which, at each step, every robot may move to an adjacent vertex, subject to the constraint that no vertex or edge is occupied by more than one robot at the same time. The two by far most prominent and well-studied objectives for \gcmp\ are the \emph{makespan} (i.e., the total number of time steps)~\cite{banfi,EibenGK23,FioravantesKKMO24,DeligkasEGK025} and the total distance traveled by the robots~\cite{PapadimitriouRST94,GeftHalperin,icalp,DeligkasEGKLR26}. In this work, we focus exclusively on the latter objective. 

From an algorithmic standpoint, CMP has attracted substantial attention due to its intrinsic computational difficulty and its close connections to graph search, scheduling, and constraint satisfaction. The problem is further motivated by numerous real-world applications, including automated warehouse systems~\cite{kiva,0001TKDKK21}, traffic management and control~\cite{MorrisPLMMKK16}, and robotics~\cite{VelosoBCR15}. A prominent example is the coordination of hundreds or thousands of mobile robots in Amazon fulfillment centers, where systems such as the Kiva robots rely on efficient and reliable multi-robot path planning to move inventory pods while avoiding collisions~\cite{kiva}. The \gcmp{} problem is known to be \NP-hard~\cite{nphard1,nphard2,lavalle} and was posed as the SoCG 2021 Computational Challenge~\cite{socg2021}. 

While the classical complexity of CMP was already studied in the 1990s, investigations from the more refined perspective of parameterized complexity~\cite{DowneyFellows13,CyganFKLMPPS15} have only begun recently. Most existing work considers the number $k=\mathcal{R}$ of robots as the natural parameter. In particular, in 2023 CMP was shown to be fixed-parameter tractable\footnote{That is, admits an algorithm running in time $f(k)\cdot |V(G)|^{\bigoh(1)}$, for some computable function $f$ of $k$.} on full rectangular grids~\cite{EibenGK23}. One year later, a follow-up work established the fixed-parameter tractability w.r.t.\ $k$ and the treewidth $\tw(G)$ of the underlying graph $G$, without restricting to grids~\cite{icalp}. In fact, that work considered a generalization of CMP, termed \cmpl, in which a subset of the robots do not have destinations. In both works, the fixed-parameter (in-)tractability of CMP on subgrids and on planar graphs are identified as significant open problems in the field.

\subparagraph{Contributions.}
In most practical settings, robot motion takes place in a continuous geometric environment. While there exists a large and largely separate body of work on heuristics and approximation algorithms for robot motion in the plane~\cite{heuristic3,heuristic2,heuristic1,heuristic4,lavalle,lavalle1,halprinunlabeled,demaine,survey,alagar,sharir,sharir1,sharir2}, a common and highly successful paradigm is to discretize the environment, yielding a graph on which the CMP problem is solved. Typically, the free space of a polygonal region is discretized using grids, after which CMP algorithms are applied to the resulting graph~\cite{grid2,grid1}. 

This paradigm underlies a wide range of CMP algorithms, including classical grid-based formulations, conflict-based search and its variants, and hybrid methods that combine sampling-based motion planning with discrete multi-agent coordination~\cite{banfi,grid1,grid2,heuristic3,heuristic2,heuristic1,heuristic4,lavalle,lavalle1,halprinunlabeled,survey,FioravantesKKMO24}. Discretization has also played a key role in recent approximation algorithms for the problem~\cite{SODA26new}. Consequently, the graphs arising in many CMP applications can naturally be viewed as discretizations of simple polygons.

In this paper, we establish the fixed-parameter tractability of \textsc{Coordinated Motion Planning}, as well as its aforementioned generalization \cmpl, on graphs arising from discretizations of simple polygons (see Figure~\ref{fig:discretizedpolygon} for an illustration and Section~\ref{sec:prelims} for a formal definition).

\begin{restatable}{theorem}{thmmain}
    \label{thm:main}
    \cmpl (and hence also \gcmp) is fixed-parameter tractable w.r.t.\ $k$ on discretizations of simple polygons.
\end{restatable}

\begin{figure}[htbp]
\center
\includegraphics{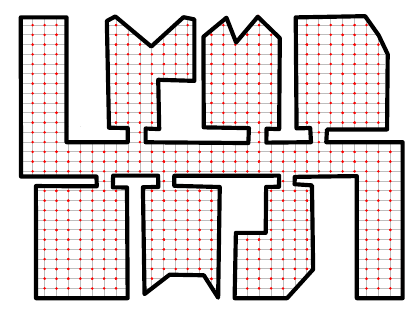}
\caption{An illustration of a discretization of a simple polygon.}
\label{fig:discretizedpolygon}
\end{figure}

Theorem~\ref{thm:main} extends existing tractability results beyond full rectangular grids to a broader and more realistic class of graphs. Although the gap between discretized simple polygons and full rectangular grids may appear modest at first glance, the presence of a potentially highly intricate boundary fundamentally changes the problem: any algorithm must optimally navigate robots through narrow corridors and  tight geometric constraints. In this setting, both core ingredients of the fixed-parameter algorithm for full rectangular grids---namely, bounding the number of turns~\cite[Theorem 17]{EibenGK23} and the use of ILP formulation~\cite[Theorem 18]{EibenGK23})---break down completely.

A second ``direct'' approach to proving Theorem~\ref{thm:main} would be to build on the recently established fixed-parameter tractability of \gcmp{} w.r.t.\ $k$ plus the treewidth of $G$~\cite{icalp}. To this end, one might hope to apply the \emph{irrelevant vertex technique}~\cite{AdlerKKLST17} to iteratively identify and remove vertices from $G$ until obtaining an equivalent instance whose underlying graph has bounded treewidth. 

While the irrelevant vertex technique has been successfully applied to the related \textsc{Planar Disjoint Paths} problem~\cite{AdlerKKLST17}, it does not preserve shortest paths, a property that is crucial for CMP. Indeed, fixed-parameter tractability for \textsc{Planar Shortest Disjoint Paths} was established only in a very recent breakthrough~\cite{SODA26paths}. Moreover, the approach developed there does not extend to \gcmp{}, as it relies fundamentally on the requirement that the paths be disjoint.

Instead of identifying and removing irrelevant vertices, we prove Theorem~\ref{thm:main} using a novel treewidth-reduction technique.
Our approach does \emph{not} necessarily yield a subgraph of the original graph; rather, it modifies the original graph by introducing ``weaving'' gadgets.
This creates a new graph that is neither planar nor smaller than the original graph; however, it produces a graph with bounded treewidth.

Our technique conceptually rests on the following four main ingredients:
\medskip
\begin{enumerate}
\item a decomposition of $G$ into so-called \emph{sectors};

\item a \emph{flattening} step that identifies and removes parts of $G$ that are ``too far'' to be useful; 

\item a \emph{weaving} reduction step that replaces certain interior parts of sectors with a collection of vertex-disjoint paths; and

\item a proof that exhaustive application of these steps results in a graph of bounded treewidth.
\end{enumerate}
\medskip

We then solve \gcmp{} using the fixed-parameter algorithm parameterized by $k+\tw(G)$~\cite{icalp}.
An overview of these four ingredients, together with a discussion of their relationship to previous work, is provided in Section~\ref{sec:overview}. 

\subparagraph*{Further Related Work.}
In addition to the travel-distance minimization objective studied in this paper, coordinated motion planning has been extensively investigated from a parameterized complexity perspective under the objective of minimizing the makespan~\cite{DeligkasEGK025,EibenGK23,FioravantesKKMO24,FKMMO25}.

Although the computational complexity of many fundamental coordinated motion planning variants was established several decades ago~\cite{PapadimitriouRST94}, recent years have witnessed renewed interest in these through the lens of parameterized-complexity. Beyond the earlier work on \cmpl~\cite{icalp}, this line of research has yielded parameterized analyses for settings involving high-speed robots~\cite{EGKS25SoCG}, fixed-parameter algorithms on solid grid graphs~\cite{EibenGK23}, lower-bound frameworks for tree-like and other structurally restricted graph classes~\cite{FioravantesKKMO24,FKMMO25}, as well as fixed-parameter approximation schemes for tree graphs~\cite{DeligkasEGK025}.

Importantly, several of these contributions—most notably the latter three—focus on makespan minimization in a parallel execution model, giving rise to complexity phenomena that differ markedly from those arising in energy-based objectives. For example, even deciding the existence of a constant-makespan schedule is \NP-hard on full grids~\cite{EibenGK23}. In contrast, when the objective is to minimize the total energy, constant-energy schedules can be computed in polynomial time; moreover, this result extends to fixed-parameter tractability for graph classes with bounded local treewidth~\cite[Theorem~5]{icalp}.

We also note recent work on the parameterized complexity of other variants of coordinated motion planning~\cite{DemaineHM14,halprinunlabeled}, which differ from the problem considered here in both their modeling assumptions and optimization objectives.

\section{Preliminaries}
\label{sec:prelims}
 
We use standard graph-theoretic terminology~\cite{Diestel12}. For $z \in \nat^*$, we write $[z]$ for $\{1,\dots,z\}$. 

Let $P$ be a simple polygon, that is, a non-self-intersecting polygon without holes, embedded in the Euclidean plane and containing a unit-length grid $H$. The \emph{discretization} of $P$ is the subgraph of $H$ consisting of all vertices and edges that are fully contained in $P$. A finite subgrid $G$ is called a \emph{discretized polygon} if it is the discretization of some simple polygon $P$\footnote{We note that the notion of discretized polygons considered here can also be seen as a graph representation of simply connected \emph{polyominos}.}. 

We assume that every discretized polygon is equipped with a fixed embedding that specifies whether each edge is vertical or horizontal. Accordingly, we say that a vertex $v$ is a $\rightarrow$-neighbor of a vertex $w$ if $v$ can be reached from $w$ by traversing an edge that travels rightward, and analogously for $\{\uparrow,\downarrow,\leftarrow\}$.

Given a path $P$ in a discretized polygon $G$, a vertex $v\in V(P)$ is called a \emph{bend} if it has two incident edges in $P$ and these edges do not leave $v$ in opposite directions. If $P$ has no bends, we call it a \emph{straight path}. Moreover, a vertex $v\in V(G)$ is called a \emph{boundary vertex} if it has degree at most $3$; intuitively, such vertices lie near the boundary of the polygon $P$ underlying $G$.

\subparagraph{The Motion Planning Problem.}
In our problem of interest, we are given a graph 
$G$ together with a set $\R=\{R_1, R_2, \ldots, R_k\}$ of $k$ robots. The set $\R$ is partitioned into two subsets $\M$ and $\F$; the former, $\M$, consists of robots with prescribed destinations, while $\F$ contains the remaining ``free'' robots. Each robot $R_i \in \M$ is associated with a starting vertex $s_i$ and a destination vertex $t_i$ in $V(G)$, whereas each robot $R_i \in \F$ is associated only with a starting vertex $s_i \in V(G)$.
The elements of the set $\{s_i \mid i \in [k]\} \cup \{t_i \mid R_i \in \M\}$ are called \emph{terminals}. We assume that all starting vertices $s_i$ are pairwise distinct and that all destination vertices $t_i$ are pairwise distinct.

Intuitively, at each time step a robot may either move to an adjacent vertex  or remain at its current vertex, and all robots may move simultaneously. To formalize this, a \emph{route} for a robot $R_i$ be a tuple $W_i=(u_0,\ldots,u_q)$ 
of vertices in $G$ satisfying: (1) $u_0=s_i$ and $u_q=t_i$, and (2) for all $j \in [q]$, either $u_{j-1}=u_j$ or $u_{j-1}u_j \in E(G)$. 

We use a discrete time interval $[0,q]$, where $t \in \nat$, to index the sequence of robot movements; at each time step $x \in [0,q]$, every robot either remains stationary or moves to an adjacent vertex.
  
Two routes $W_i=(u_0,\ldots,u_q)$ and $W_j=(v_0,\ldots,v_q)$, where $i \neq j \in [k]$, are said to be \emph{non-conflicting} if (i) for all $r \in \{0,\ldots,q\}$, $u_r \neq v_r$, and (ii) there does not exist an $r \in \{0,\ldots,t-1\}$ such that $v_{r+1}=u_r$ and $u_{r+1}=v_r$. Otherwise, the routes $W_i$ and $W_j$ are said to \emph{conflict}. Intuitively, two routes conflict if, at the same time, the corresponding robots either occupy the same vertex or traverse the same edge in opposite directions. 
A \emph{schedule} $\SSS$  is a set of pairwise non-conflicting routes $W_i$, $i \in [k]$, defined over a common time interval $[0,q]$. The (\emph{traveled}) \emph{length} of a route (or of its associated robot) in $\SSS$ is the number of time steps $j$ such that $u_j \neq u_{j+1}$. The \emph{total traveled length} of a schedule is the sum of the lengths of all its routes; this quantity is commonly referred to as the \emph{energy} in the literature (see, e.g.,~\cite{socg2021}). 
 
Using the terminology introduced above, we can formally define the problem 
of interest.

 \begin{problem}[]{{\cmpl}}
\Input & A tuple $(G, \R=(\M,\F), k, \ell)$, where $G$ is an undirected graph, $k, \ell \in \nat$, and $\R=\{R_i \mid i \in [k]\}$ is a set of robots partitioned into sets $\cal M$ and $\cal F$, where each robot in $\cal M$ is given as a pair of vertices $(s_i, t_i)$  and each robot in $\cal F$ as a single vertex $s_i$.\\
 \Prob & Is there a schedule for $\R$ of total traveled length at most $\ell$?
\end{problem}

The {\sc Coordinated Motion Planning Problem} (\gcmp{}) is the restriction of \cmpl{} to instances with $\F= \emptyset$. Throughout the paper, we use \dcmp\ as shorthand for \gcmp\ restricted to graphs which are discretized polygons. We note that our result (Theorem~\ref{thm:main}) is constructive: it also outputs a schedule of length at most $\ell$, if one exists.

\iflong
\subparagraph*{Parameterized Complexity.}
 
 A {\it parameterized problem} $Q$ is a subset of $\Omega^* \times
\nat$, where $\Omega$ is a fixed alphabet. Each instance of $Q$ is a pair $(I, \kappa)$, where $\kappa \in \nat$ is called the {\it
parameter}. A parameterized problem $Q$ is
{\it fixed-parameter tractable} (\FPT)~\cite{CyganFKLMPPS15,DowneyFellows13,FlumGrohe06}, if there is an
algorithm, called an {\em \FPT-algorithm},  that decides whether an input $(I, \kappa)$
is a member of $Q$ in time $f(\kappa) \cdot |I|^{\bigoh(1)}$, where $f$ is a computable function and $|I|$ is the input instance size.  The class \FPT{} denotes the class of all fixed-parameter tractable parameterized problems. We denote by \emph{\FPT-time} running time of the form $f(\kappa) \cdot |I|^{\bigoh(1)}$. We refer to~\cite{CyganFKLMPPS15,DowneyFellows13} for more information on parameterized complexity. 

\subparagraph*{Treewidth.}
The treewidth of a graph is defined via the notion of tree decompositions as follows:

\begin{definition}[{\bf Tree decomposition}]\label{def:treewidth}
{\rm A \emph{tree decomposition} of a graph $G$ is a pair $(T,\beta)$ of a tree $T$ and $\beta: V(T) \rightarrow 2^{V(G)}$,
such that: 
\begin{itemize}
\item $\bigcup_{t \in V(T)} \beta(t) = V(G)$, 
\item for any edge $e \in E(G)$, there exists a node $t \in V(T)$ such that both endpoints of $e$ belong to $\beta(t)$, and
\item the \emph{interpolation property}: for any vertex $v \in V(G)$, the subgraph of $T$ induced by the set $T_v = \{t\in V(T): v\in\beta(t)\}$ is a tree.
\end{itemize}
The {\em width} of $(T,\beta)$ is $\max_{v\in V(T)}\{|\beta(v)|\}-1$. The {\em treewidth} of $G$ is the minimum width of a tree decomposition of $G$.}
\end{definition}

Let $(T,\beta)$ be a tree decomposition of a graph $G$. We refer to the vertices of $T$ as \emph{nodes}. We assume throughout that $T$ is a rooted, which induces the usual parent-child and ancestor-descendant relationship among its nodes. A \emph{leaf} node or a \emph{leaf} of $T$ is a node with degree exactly one in $T$ that is not the root. All nodes that are neither the root nor leaves are called \emph{non-leaf} nodes. 

For each node $t \in T$, The set $\beta(t)$ is called the \emph{bag} at $t$.
  For two nodes $u,t \in T$, we say that $u$ is a {\em descendant} of $t$, denoted $u \preceq t$, if $t$ lies on the unique path connecting $u$ to the root. Note that every node is its own descendant; if $u\preceq t$ and $u\neq t$, we write $u\prec t$.
  
 For a tree decomposition $(T,\beta)$, we further define a mapping $\gamma:V(T)\to 2^{V(G)}$ by $\gamma(t)=\bigcup_{u\preceq t} \beta(u)$.

 \fi

\section{Technical Overview}
\label{sec:overview}

In this section, we present an overview of our fixed-parameter algorithm for \dcmp, parameterized by the number $k$ of robots (Theorem~\ref{thm:main}).

\subsection{Preprocessing: Initial Set-Up}
\label{sec:preprocess}
The starting point for our work is a preprocessing subroutine developed in recent work by Deligkas, Eiben, Ganian, Kanj and Ramanujan~\cite[Theorem 15---full version]{icalp}. This subroutine can be viewed as a fixed-parameter Turing-reduction which solves any instance of \cmpl\ by making calls to structurally ``simpler'' instances. In particular, in the resulting instances we only seek schedules in which every robot travels at most $c k^5$ steps, where $c$ is a constant, beyond its shortest-path distance; that is, each robot has \emph{travel slack} at most $\lambda=ck^5$. 

Moreover, we may assume that every robot in the new instance has a prescribed destination, allowing us to focus exclusively on \gcmp. The reduction preserves the property of being a discretized polygon. Consequently, in all subsequent steps we restrict our attention to  solving \dcmp\ with the above bound $\lambda$ on the travel slack.

\subsection{Step 1: Sectors and Sector Graphs}\label{subsec:sectors}
To exploit the geometry induced by discretized polygons, we adapt the notion of \emph{sectors} recently introduced in the setting of fixed-parameter extension algorithms for orthogonal (geometric) drawings~\cite{BhoreGKMN23}. Specifically, we associate with each vertex a \emph{bend vector} that records, for every terminal and every directional axis $d\in \bidirections$, the minimum number of bends on any shortest path arriving from direction $d$. Connected subgraphs whose vertices share the same bend vector form a \emph{sector}. These sectors partition the graph into regions whose vertices exhibit similar behavior with respect to the number of turns required to reach all terminals.

We then define the \emph{sector graph} $\Gamma$, whose vertices correspond to the sectors of $G$ and whose edges capture adjacencies between sectors. A key contribution of our work is establishing the following structural insights about $\Gamma$, which will later be used to bound the treewidth of $G$ after exhaustive application of our reduction rules.
     
\noindent \textbf{Insight 1.}
 The treewidth of $\Gamma$ is bounded by a function of $k$\iflong ~(Theorem~\ref{thm:sectortw})\fi.\\
\textbf{Insight 2.} 
 Every straight path in $\Gamma$ intersects at most $12k$ sectors\iflong  ~(Lemma~\ref{lem:pathsectors})\fi.\\
\textbf{Insight 3.}
 Every sector has at most $8$ ``non-trivial'' neighbors in $\Gamma$\iflong ~(Lemma~\ref{lem:fewnonclean})\fi.

\begin{remark*}
\emph{
In prior work introducing sectors}~\cite{BhoreGKMN23} \emph{in the geometric setting, $d$ ranged over $\directions$ and the anchors defining the sectors were required to lie on the boundary of the polygon. In our setting, the change to $\bidirections$ is necessary; without it, several of the lemmas underpinning the treewidth proof \iflong(including, e.g., Lemma~\ref{lem:sectortree})\fi would not hold.
 This is one---though not the primary---reason why the  treewidth bound from}~\cite[Theorem 22]{BhoreGKMN23} \emph{in the geometric setting cannot be reused.
 The main reason we cannot directly adapt the proofs from the previous paper is that those rely on prior pruning steps that have removed ``dangling'' sectors} (\emph{which cannot be replicated for \dcmp}).
\end{remark*}

\begin{figure}
     \centering
     \includegraphics[page=10,width = 0.65\linewidth]{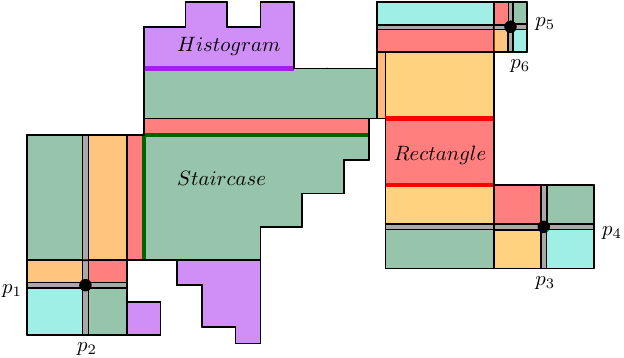}\includegraphics[page=12, width=0.3\linewidth]{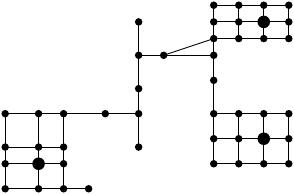}
     \caption{Left: An illustration of the sectors of a discretized polygon w.r.t.\ three terminals defining six ports $p_1, \ldots, p_6$. Three of the sectors have highlighted their baselines, showcasing the three non-trivial types of a sector. 
     Notice that terminals form their own ``degenerate'' singleton sectors.
     Right: The corresponding sector graph $\Gamma$ as defined by the sectors.}
     \label{fig:sectorpartition}
 \end{figure}
As a secondary structural insight, we show that each nontrivial sector contains a canonical \emph{baseline}: a straight path from which the entire sector is reachable by straight extensions. Depending on the baselines present, we obtain a small taxonomy (see Figure~\ref{fig:sectorpartition}).
\begin{itemize}
    \item \textbf{Histogram sectors}, with a single baseline.
    \item \textbf{Staircase sectors}, with two orthogonal baselines.
    \item \textbf{Rectangle sectors}, with two parallel baselines forming a full subgrid.
\end{itemize}

Crucially, there exists a shortest path from a terminal to a vertex in a nontrivial sector which intersects one of the sector's baselines. This fundamental geometric constraint enables us to apply localized reductions while preserving shortest-path feasibility.

\subsection{Steps 2-3: The Reductions}\label{subsec:reductions}
Building on the taxonomy and structural insights described above, we design reduction rules tailored to each sector type. These rules simplify the instance while preserving full equivalence with respect to the existence of a feasible schedule. 

\subparagraph{Reduction for Histogram Sectors.}
For histogram sectors, we observe that any vertex that lies more than $\frac{\lambda}{2}$ steps away from the baseline cannot be part of any route in a feasible solution. We therefore ``flatten'' histogram sectors by removing all such vertices, as illustrated in Figure~\ref{fig:red_histsummary}\iflong and formalized in Lemma~\ref{lem:histogramreduction}\fi.  

\begin{figure}[h]
    \centering
    \includegraphics[page=1]{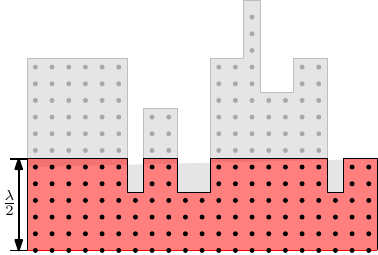}
    \caption{Illustration of the reduction for histogram sectors. The gray vertices are removed from the graph.}
    \label{fig:red_histsummary}
\end{figure}

\subparagraph*{Reduction for Staircase Sectors.}
A staircase sector $S$ contains two orthogonal baselines $P_S^1$ and $P_S^2$ that separate all terminals from the rest of the sector (and from any sectors attached to the interior). We show that if there exists a hypothetical solution that routes a robot out of a ``buffer zone'' of radius $(k+1)(\lambda+1)$ around these baselines into a staircase, then a sequence of rerouting arguments yields an alternative solution in which all robots stay within the buffer zone. This, in turn, allows us to safely remove all vertices outside the buffer zone\iflong (cf. Lemma~\ref{lem:staircasered} and Reduction Rule~\ref{red:stair})\fi. An illustration is provided in Figure~\ref{fig:red_staircasesummary}.

\begin{figure}
    \centering
    \includegraphics[page=13]{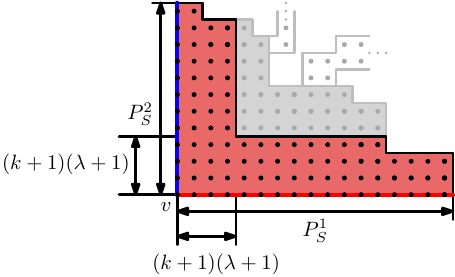}
    \caption{Illustration of the reduction for staircase sectors. The red region highlights a staircase sector $S$. Gray vertices indicate the vertices removed by the application of the reduction.}
    \label{fig:red_staircasesummary}
\end{figure}

\subparagraph*{Reduction for Rectangle Sectors.}
  Rectangle sectors require a more intricate reduction. We show that any feasible schedule can be transformed into a \emph{semi-canonical} form in which all turns and waiting steps occur within a bounded-width frame along the boundary of the sector. As a consequence, the \emph{second frame} $S^*$---a central rectangular region in a rectangle sector $S$---contains no turns or waiting vertices/steps (see Figure~\ref{fig:red_rectsummary}). The proof is non-trivial and relies on a sequence of buffer-shifting arguments, together with a previously established upper bound $\mu\in k^{\bigoh(k^2)}$ on the number of turns in a hypothetical solution, for full rectangular grids~\cite{EibenGK23}.

\begin{figure}
    \centering
    \includegraphics[page=3]{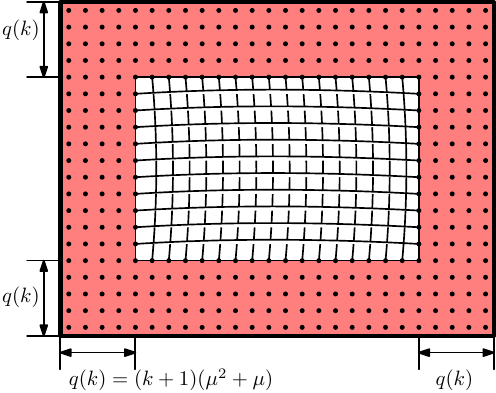}
    \caption{Illustration for reducing the rectangle sectors. The red part is the part of the rectangle sector $S$ that is not inside the second frame $S^*$. The horizontal and vertical lines connecting the opposite sides of the second frame $S^*$ do not intersect each other. }
    \label{fig:red_rectsummary}
\end{figure}

We then remove all internal vertices of $S^*$ and replace them with degree-2 \emph{weaving paths} connecting pairs of opposite boundary vertices\iflong, as formalized in Reduction Rule~\ref{red:rect}\fi. These paths preserve the relative distances relevant to robot routing. Although this ``reduction'' results in an instance with more vertices and destroys planarity, we later show that it makes tangible progress toward reducing the treewidth of $G$.

\subsection{Step 4: From Reductions to Bounded Treewidth}\label{subsec:tw_bound}
By exhaustively applying the reductions outlined above, we obtain a graph in which every sector either has small width or small height, or can be decomposed into at most four components of small width or height, together with a set of additional degree-2 paths. At this stage, it is not difficult to show that each sector individually has bounded treewidth. The main technical challenge lies in proving that the entire graph---where sectors are interconnected along the edges of the sector graph $\Gamma$---has bounded treewidth. 

\begin{remark*}
\emph{
A treewidth bound on $G$ does \emph{not} follow directly from the fact that $\Gamma$ has bounded-treewidth and its nodes represent connected components of $G$ with bounded-treewidth. For instance, consider a long sequence of ``flat but wide'' sectors stacked on top of each other, i.e., connected along a single long path $\Gamma$. A graph $G$ with this structure could have arbitrarily large treewidth, even though each individual sector and the sector graph $\Gamma$ are of bounded treewidth. While Insight~2 excludes this particular configuration, more complex arrangements of sectors could arise and be chained together, requiring a more careful argument.
 }
\end{remark*}

Our treewidth-bounding argument for $G$ leverages a bounded-width tree-decomposition $\T$ of the sector graph $\Gamma$ (as per Insight 1) to guide the construction of a tree-decomposition 
of $G$. We first preprocess $G$ according to $\Gamma$ to safely detach all sectors that are ``trivial'' in the sense of Insight 3; this reduces the task to bounding the treewidth of
$G$ under the assumption that $\Gamma$ has degree at most $8$. Crucially, this implies that the $12k$-th power graph $\Gamma^+$---the supergraph of $\Gamma$ that contains an edge between each pair of sectors of distance at most $12k$---also has bounded treewidth. Let $(\T^+,\beta^+)$ denote a tree-decomposition witnessing this bound.

Next, for any arbitrarily chosen sector in $\Gamma^+$, we show that it is possible to carefully identify a set of \emph{special} vertices in the sector (see Figure~\ref{fig:initial_special_verticessummary}); the number of special vertices per sector is upper-bounded by a function of $k$, and their positions allow them to act as ``semi-separators''. Using the special vertices, we construct a tree-decomposition of $G$ via a two-step approach.
\medskip
\begin{figure}[h]
    \centering
    \includegraphics[width=\linewidth,page=4]{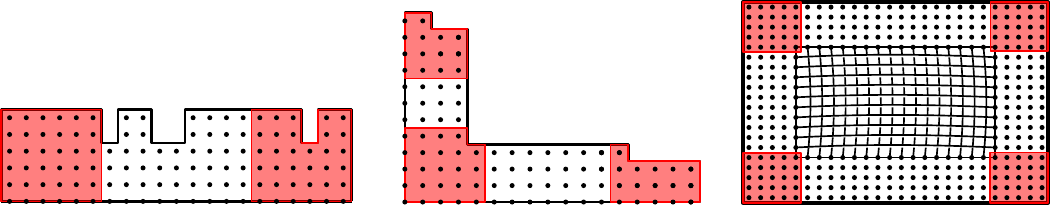}
    \caption{Illustration of initial special vertices for the three types of sector. For each sector $S$, initial special vertices split the sector into connected components such that each connected component has bounded treewidth and is connected either only to components in the horizontal neighbors of $S$ or only to components in the vertical neighbors of $S$. Further special vertices are added during the course of our procedure, but we maintain a bound on their number in each sector.}
    \label{fig:initial_special_verticessummary}
\end{figure}
 \begin{enumerate}
\item We first construct a tree-decomposition template $(\T^+,\beta')$ for a subgraph of $G$ as follows: for each $t\in \T^+$, $\beta'(t)$ replaces each sector in $\beta^+(t)$ with its corresponding special vertices. Crucially, the size of each bag in this template can be upper-bounded by the product of (a) the maximum number of special vertices in a sector and (b) the width of $(\T^+,\beta^+)$.

\item Next, we show that every connected component $C$ in $G-\bigcup_{t\in V(\T^+)}\beta'(t)$ is a subgraph formed by sectors intersected by the same straight path. Using Insight 2, we can show that $C$ has bounded treewidth and that its entire neighborhood $N(C)$ is contained in the bag $\beta'(t)$ for some node $t\in V(\T^+)$.
 This allows us to append a tree-decomposition of $C\cup N(C)$, rooted at the bag containing $N(C)$, to $t$. 
 \end{enumerate}

\subsection{Putting Everything Together}
Having introduced the technical components of our approach, we now combine them to derive an algorithm that proves Theorem~\ref{thm:main}.
  Let $\I= (G, \R, k, \ell)$ be an instance of \cmpl{} on discretized polygons. We solve $\I$ via the following steps.
\medskip
\begin{enumerate}
    \item 
         Preprocess $\I$ following the procedure described in Subsection~\ref{sec:preprocess}
         to ensure bounded slack and to reduce \cmpl{} on discretized polygons to \dcmp. 
    \item Construct the sectors and the sector graph, as detailed in Subsection~\ref{subsec:sectors}.
    \item Exhaustively apply all sector-specific reductions described in Subsection~\ref{subsec:reductions} yielding a graph whose treewidth is bounded by a function of $k$.
    \item Apply the existing fixed-parameter algorithm for \gcmp{} parameterized by $k$ and the treewidth of the resulting graph~\cite{icalp}.  
\end{enumerate}

This completes the description of the algorithm and establishes Theorem~\ref{thm:main}. 
\iflong
In the remainder of the paper, we formalize these steps and show how they work together.

\section{Initial Setup}
Let $\I=(G,\R=(\M, \F),k,\ell)$ be an instance of \cmpl{} on discretized polygons. The starting point for our work is a preprocessing subroutine developed in recent work by Deligkas, Eiben, Ganian, Kanj, and Ramanujan~\cite[Theorem 15---full version]{icalp}\footnote{We note that the theorem applies to \cmpl{} on general graphs.}. This subroutine can be viewed as a fixed-parameter Turing-reduction which solves any instance of \cmpl\ by making calls to structurally ``simpler'' instances:

\begin{theorem}[\hspace{-0.001cm}\cite{icalp}]
\label{thm:fptapx}
 Let $\I=(G, \R=(\M,\F), k, \ell)$ be an instance of \cmpl{}. In \FPT{}-time, we reduce $\I$ to $p=g(k)$-many instances (for some computable function $g(k)$) $\I_1, \ldots, \I_p$
such that
\begin{itemize}
\item for all $j\in [p]$, $\I_j$ has a schedule in which the total traveled length is at most \\ $\sum_{R_i \in \M_j}\dist(s_i, t_i) + c \cdot k^5$, for some constant $c > 0$; and
\item $\I$ is yes-instance if and only if there exists a $j\in [p]$ such that $\I_j=(G_j, \R_j=(\M_j, \F_j), k_j, \ell_j)$ is yes-instance. Moreover, the instance  $\I_j$ satisfies that $k_j \leq k$, $\ell_j \leq \ell$, and $G_j$ is a subgraph of $G$ obtained by possibly removing degree-2 vertices from $G$. 

\end{itemize}
    \end{theorem}

Observe that applying Theorem~\ref{thm:fptapx} to $\I$, since $G$ is a discretized polygon and $G_j$ is obtained from $G$ by possibly removing degree-2 vertices from $G$, it follows that $G_j$ itself is a disjoint union of discretized polygons.
 If $G_j$ has multiple connected components, we apply the algorithm independently to each component. Henceforth, we assume that our input instance $\I$ has been preprocessed using Theorem~\ref{thm:fptapx} and satisfies $\ell\le \sum_{R_i\in \M}\dist(s_i, t_i)+ c\cdot k^5$, for some fixed constant $c> 0$. Consequently, no robot in $\M$ travels more than $c \cdot k^5$ steps beyond the shortest-path distance between its starting and destination points. We call the excess length that a robot $R_i$ travels beyond $\dist(s_i,t_i)$ in a schedule, the \emph{travel slack} of $R_i$. Moreover, the final position of each robot in $\F$ lies within distance at most $c \cdot k^5$ from its starting vertex.

Since $G$ is a subgrid and thus has maximum degree at most $4$, for each robot in $\F$ we may branch into at most $4^{c k^5}$ possibilities to guess its final position in a potential solution. The total number of branches over all robots in $\F$ is therefore at most $4^{c k^6}$. After this branching step—which increases the running time by a multiplicative function of $k$—the instance $\I$ is reduced to an instance of \cmpl{}, in which every robot has a destination and that satisfies the statements of Theorem~\ref{thm:fptapx}.

Thus, without loss of generality, we may restrict our attention to solving the \dcmp{} problem. Moreover, we may assume that the input instance of \dcmp{} under consideration satisfies the statement of Theorem~\ref{thm:fptapx}.

The paper~\cite{EibenGK23} studied \gcmp{} restricted to full rectangular grids. It was shown that if an instance $\I=(G,\R,k,\ell)$ admits a solution, then it admits a solution in which each robot makes at most $\tau(k)$ turns, for some function $\tau$ depending only on $k$; thus was crucial for their result, and will also be an important component for dealing with large rectangular sectors here. More precisely:
 \begin{theorem}[\hspace{-0.001cm}\cite{EibenGK23}]
\label{thm:boundedturnsdistance}
  If $\I=(G, \R, k, \ell)$ is a \YES-instance of \gcmp\ on full rectangular grid, then $\I$ has a valid schedule $\SSS$ in which each route makes at most $\bigoh(\tau(k)^{2k+1})$ turns, where $\tau(k)  =3k^k (\sigma(k)+1)+\sigma(k)$, and $\sigma(k)=4k^2$.
\end{theorem}

Let $\I=(G,\R,k,\ell)$ be an instance of \dcmp{}  and suppose that $\I$ admits a solution $\SSS$. Let $B$ be a subgraph of $G$ that induces a full rectangular grid. Assume that there exists a time interval $T$ during which the set of robots contained in $B$ remains unchanged.

We decompose $\SSS$ into $\SSS^-$, $\SSS_T$, and $\SSS^+$, where $\SSS^-$ is the restriction of $\SSS$ to the time interval $T^-$ preceding $T$, $\SSS_T$ is its restriction to $T$, and $\SSS^+$ is its restriction to the time interval $T^+$ following $T$.

By applying  Theorem~\ref{thm:boundedturnsdistance} to the interval $T$ and the subgraph $B$, we obtain another solution $\SSS^*$ to $\I$ with the following properties. The schedule $\SSS^*$ coincides with $\SSS$ throughout the time interval $T^-$. During interval $T$, robots outside $B$ follow the same routes as in $\SSS_T$, while robots inside $B$ follow an alternative schedule that routes them from their positions at the beginning of $T$ to their positions at the end of $T$, with each robot making at most $\bigoh(\tau(k)^{2k+1})$ turns. These two schedules may have different durations; in this case, robots in the faster schedule remain stationary until the slower schedule completes. After both schedules have finished, $\SSS^*$ continues identically to $\SSS^+$.

Note that, since the travel slack of each robot is $\bigoh(k^5)$, it follows that each robot can enter or leave the subgraph $B$ at most $\bigoh(k^5)$ times. Consequently, the schedule $\SSS$ can be partitioned into $\Oh(k^6)$ time intervals, each of which has the property that the set of robots contained in $B$ remains unchanged. By the discussion above, this implies that we may assume the existence of a solution in which each robot makes at most $\Oh(k^6 \cdot \tau(k)^{2k+1})$ turns while it is inside $B$ over the entire schedule.

Observe that bounding the number of turns made by the robots also bounds the number of grid lines along which the robots move; we refer to such grid lines as \emph{important}. Based on the above theorem, it was further shown that if an instance $\I$ admits a solution, then it admits a solution in which robots wait (i.e., do not move during a time step) only at vertices that can be reached by a straight path of length at most $k$ from the intersection of two important lines (see Claim~1 in Theorem~14 of~\cite{EibenGK23}). 

This observation implies an upper bound of
$\Oh(4k^3(k^6\cdot \tau(k)^{2k+1}+2)^2)=k^{\Oh(k^2)}$
on the number of vertices at which robots may wait.

We formalize the above discussion in the following assumption:

\begin{assumption}
\label{ass:canonical}
Let $\I=(G, \R, k, \ell)$ be an instance of \dcmp{}. If $I$ is a \YES-instance then there is a solution to $\I$ satisfying the following:

\begin{itemize}

    \item[(i)] No robot $R_i$ has travel slack more than $\lambda = ck^5$. 

    \item[(ii)] For each robot $R_i \in \R$  and for each rectangle (i.e., full rectangular grid) $B$ in $G$, the number of turns that $R_i$ makes in $B$ plus the number of vertices in $B$ at which it waits
    is at most $\mu=k^{\Oh(k^2)}$. 
\end{itemize}
\end{assumption}

\section{Sectors and the Sector Graph}
\label{sec:sectors}

Let $\textsf{Terminals}$ denote the set of all starting vertices and destination vertices appearing in $\R$, and fix an arbitrary ordering of $\textsf{Terminals}$. 
A \emph{port} is a pair $(a,d)$, where $a \in \textsf{Terminals}$ and $d \in \bidirections$. Observe that the number $\kappa$ of ports in an instance of \dcmp{} is upper-bounded by $4 \cdot |\R|$.
The \emph{bend distance} between a vertex $v \in V(G)$ and a port $(a,d)$, denoted $\bdist(v,(a,d))$, is defined as the minimum integer $i$ such that there exists a $v$--$a$ path in $G$ with exactly $i$ bends that reaches $a$ by arriving from direction $d$ (see Figure~\ref{fig:bend_distances}).

 \begin{figure}
    \centering
    \includegraphics[page=11]{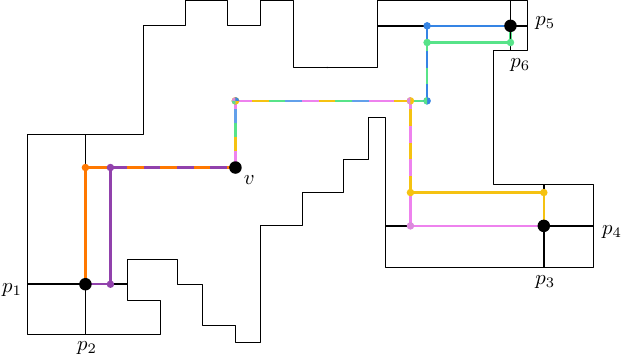}
    \caption{An example of a discretized polygon with three vertices in the set \textsf{Terminals} and hence six ports. For each port $p_i = (a_i,d_i)$, there is an example of a path between vertices $v$ and $a_i$ arriving from direction $d_i$ with the minimum number of bends. The bend vector of the vertex $v$ is $(2, 1, 4, 3, 3, 4)$. }
    \label{fig:bend_distances}
\end{figure}

Let $p_1,\ldots,p_{\kappa}$ be an arbitrary lexicographic ordering of the ports, that is, the ports associated with each vertex in $\textsf{Terminals}$ appear consecutively; we call this a \emph{port sequence}. The \emph{bend vector} of a vertex $v \in V(G)$, denoted $\bvect(v)$, is the $\kappa$-tuple $\bigl(\bdist(v,p_1), \ldots, \bdist(v,p_{\kappa})\bigr)$.

\begin{definition}
\label{def:sectors}
Given a graph $G$ and a port sequence $\psi$, a \emph{sector} is a maximal connected subgraph of $G$ whose vertices all have the same bend vector with respect to $\psi$. Unless stated otherwise, for an instance $\I=(G,\R=(\M,\F),k,\ell)$ of \dcmp{}, we assume $\psi = (p_1,\ldots,p_{\kappa})$.
\end{definition}

Observe that the sectors form a partition of $V(G)$. Since all vertices in a sector $S$ have identical bend distances, we may simply write $\bdist(S,\cdot)$ to denote this common value. The adjacency relationships between sectors can be captured by the \emph{sector graph}, which we define next.

\begin{definition}
\label{def:sectorgraph}
The \emph{sector graph} $\Gamma$ of a discretized polygon $G$ and a port sequence $\psi$ is the graph whose vertices are the sectors of $G$ w.r.t.\ $\psi$, and where two sectors $A,B$ are adjacent in $\Gamma$ if and only if there exists an edge of $G$ with precisely one endpoint in each of $A$ and $B$.
\end{definition}

Note that $\Gamma$ is planar; moreover, the underlying grid $G$ induces a concrete planar embedding of $\Gamma$. Our goal is to establish an upper bound on the treewidth of $\Gamma$, analogous to the bound obtained for sector graphs arising from continuous spaces with boundary ports.

To this end, it is useful to consider the sequence of sector graphs obtained by progressively incorporating ports. Specifically, for each $i \in [\kappa]$, let $\Gamma_i$ denote the sector graph defined with respect to the single port $(p_i)$, and let $\Gamma_{\leq i}$ denote the sector graph defined with respect to the port sequence $(p_1,\ldots,p_i)$. 
Before we proceed to establishing the properties of sectors and the sector graph, we introduce a final notion that is closely related to them. 

\begin{definition}[$d$-baseline]
\label{def:baseline}
A straight path $P$ is a $d$-\emph{baseline} for a sector $S\in V(\Gamma_{\leq i})$ and a direction $d\in \directions$ if:
 \begin{itemize}
\item $P$ is contained in $S$;
 \item every vertex of $S$ can be reached by a straight path starting from $P$ and following the direction $d$; and
\item there exists a sector $S'\in V(\Gamma_{\leq i})$ such that 
\begin{enumerate} 
\item every $\bar d$-neighbor of each vertex of $P$ lies in $S'$, where $\bar d$ is the direction opposite to $d$, and 
\item for at least one port $p_j$ with $j\in [i]$, $\bdist(S,p_j)=\bdist(S',p_j)+1$.
\end{enumerate}
\end{itemize}
\end{definition}

When the direction $d$ is not specified, we refer to $P$ simply as a 
 \emph{baseline}; when dealing with baselines for distinct sectors, we use the notation $P_S$ to denote a baseline of $S$.
 We now show that each sector graph defined with respect to a single port is acyclic and that nearly all of its sectors contain a baseline. The acyclicity part of this result can be seen as a discrete analogue of the corresponding acyclicity property for sector graphs of continuous spaces defined with respect to $\directions$~\cite{BhoreGKMN23}.

\begin{lemma}
\label{lem:sectortree}
 For each $i\in [\kappa]$, $\Gamma_i$ is a tree. Moreover, each sector $S\in V(\Gamma_i)$ is either:
\begin{itemize}
\item the unique sector $S_0$ with $\bdist(S_0,p_i)=0$, which consists of a maximal straight path, or
\item has precisely one neighbor $\pred(S)$ such that $\bdist(\pred(S),p_i)=\bdist(S,p_i)-1$.
\end{itemize}
In the latter case, $S$ contains a baseline $P_S$ and $P_S$ is a separator of $G$. Moreover, $S$ consists \emph{precisely} of all vertices that are reachable from $P_S$ via a straight path in one specified direction.
 \end{lemma}

\begin{proof}
We prove the lemma by induction on the bend distance to $p_i$. For the base case, note that a vertex has bend distance $0$ to $p_i$ if and only if it lies on the correspondingly oriented maximal straight path $S_0$ in $G$ containing $p_i$.

Now, assume that the lemma holds for all sectors with bend distance up to $\ell-1$. If no vertex has bend distance $\ell$, then all sectors have bend distance at most $\ell-1$, and the lemma holds. Otherwise, consider a vertex $v$ with bend distance $\ell$, as witnessed by an $\ell$-bend $v$-$p_i$ path $P_v$. The first bend on $P_v$ must occur in a sector $S'$ with bend distance $\ell-1$. Let $ab$ be the first edge on $P_v$ between $v$ and the first bend such that $a\not \in S'$ and $b\in S'$. Without loss of generality, assume $ab$ is a horizontal edge with $a$ on the left; the arguments for the other cases are symmetric. Observe that $\bdist(a,p_i)=\ell=\bdist(v,p_i)$.

By the induction hypothesis, $S'$ is either a maximal straight path (if $S'=S_0$) or every vertex in $S'$ can be reached by a straight line from a path $P_{S'}$ satisfying the conditions of the lemma---in particular, it is straight and entirely adjacent to the sector $\pred(S')$. Since $ab$ is horizontal, $P_{S'}$ cannot be vertical; otherwise $v\in S'$, which contradicts our choice of $v$. Therefore, there exists a maximal straight vertical path $P_b$ in $S'$ starting from $P_{S'}$ and containing $b$. Let $P_a$ be the maximal straight vertical path containing $a$, such that every vertex in $P_a$ is adjacent to a vertex in $P_b$. 

Now, define $S$ to be the set of all vertices reachable by straight $\leftarrow$-paths from $P_a$. Observe that every vertex of $S$ has bend distance $\ell$. To prove that $S$ is a sector and satisfies the remaining desired properties, we first establish the following claim:

\begin{claim}
\label{claim:treesep}
$P_a$ is a separator between $v$ and $S'$ in $G$.
 \end{claim}

\begin{claimproof}
Suppose, towards a contradiction, that there exists a $v$-$p_i$ path $W$ in $G-P_a$. Let $xw$ be the first edge of $W$ such that $x\in S$ and $w\not \in S$. Since, by construction, $S$ has no $\leftarrow$-neighbors and $W$ does not intersect $P_a$, $w$ must be a $\uparrow$- or $\downarrow$-neighbor of $x$.

Consider the maximal $\rightarrow$-path $P_w$ starting from $w$. Note that $P_w$ cannot intersect $S'$, because the first vertex of $S'$ encountered by $P_w$ would need to be directly adjacent to a neighbor of $P_a$ as all vertices of $S'$ are reachable by a straight line from $P_{S'}$, and the immediate predecessor of that vertex on $P_w$ would then contradict the maximality of $P_a$. Therefore, the length of $P_w$ must be shorter than the length of the $\rightarrow$-path from $x$ to $P_a$. In particular, $P_w$ ends at some boundary vertex $\beta$.

Construct a cycle $C$ in $G$ as follows: start from the edge $xw$, follow $W$ outside of $S$ until reaching the first vertex that lies in $P_v$, then follow $P_v$ back towards $S$ until reaching $S$, and finally close the cycle via an arbitrary path inside $S$. Since $G$ is a plane graph, $C$ encloses $\beta$---contradicting the fact that $G$ is a discretized polygon.
\end{claimproof}
By Claim~\ref{claim:treesep}, every vertex $w$ adjacent to $S$ must satisfy $\bdist(w,p_i)=\ell+1$. Indeed, every $w$-$p_i$ path must pass through $P_a$, necessitating the use of one more bend than the bend distance from any vertex in $S$. Hence, $S$ is a sector, and $P_S:=P_a$ satisfies all the properties of Definition~\ref{def:baseline} and the lemma.
\end{proof}

Before proceeding to the general case of $\Gamma_{\leq i}$, we use Lemma~\ref{lem:sectortree} to establish a useful ``geometric'' property of the sector graph.

\begin{lemma}
\label{lem:pathsectors}
Let $P$ be a straight path in $G$. Then $P$ intersects at most $3\cdot \kappa$ sectors.
\end{lemma}

\begin{proof}
We prove the lemma by induction on $\Gamma_{\leq i}$ for $i\in [\kappa]$, where the inductive hypothesis is that $P$ intersects at most $3\cdot i$ sectors in $\Gamma_{\leq i}$. 

For the base case of $i=1$, 
 Lemma~\ref{lem:sectortree} ensures that the path $P$ can intersect at most $3$ sectors in $\Gamma_1$. Indeed, consider a vertex $v$ on $P$ with minimum bend distance to $p_1$. The sector $S$ in $\Gamma_1$  containing $v$ must intersect $P$ in a continuous subpath $P'$, since all such vertices are reachable via a straight path from some $P_S$ by Lemma~\ref{lem:sectortree}. All remaining vertices of $P$ can reach $S$ via an additional bend, ensuring that all remaining vertices in each of the at most two (sub-)paths in $P-P'$ belong to the same sector.

For the inductive step, assume the hypothesis holds for $\Gamma_{\leq i-1}$. By the arguments in the previous paragraph, $P$ intersects at most three sectors of $\Gamma_i$---that is, there are at most two edges of $P$ that pass between sectors. Now, given a set $Z$ of sectors of $\Gamma_{\leq i-1}$ intersected by $P$, we can construct the set of sectors of $\Gamma_{\leq i}$ intersected by the same path $P$ as follows. We use the at most two edges of $P$ that pass between sectors of $\Gamma_i$ to sub-partition $Z$. This operation increases the number of sectors by at most $2$, and the lemma follows.
   \end{proof}

Next, we show that every ``non-trivial'' sector contains a baseline; this generalizes the second part of Lemma~\ref{lem:sectortree}. 
 
\begin{lemma}
\label{lem:baseline}
For every $i\in [\kappa]$ and every sector $S\in V(\Gamma_{\leq i})$, exactly one of the following holds: $S$ is a straight path in $G$, or $S$ contains a baseline.
Moreover, if $S$ is not a straight path, then for each $j\in [i]$ the following holds: the unique sector $S_j$ such that $S\subseteq S_j$ contains a baseline $P_j$ that is parallel to a baseline of $S$, and $P_j$ separates every vertex of $S$ from the terminal $p_j$.
     \end{lemma}

\begin{proof}
We prove the lemma by induction on $i$. The base case $i=1$ follows directly from the construction of $\Gamma_1$ given in Lemma~\ref{lem:sectortree}. 

For the inductive step, assume $i\geq 2$ and let $S\in V(\Gamma_{\leq i})$ be a sector that is not a straight path. Then $S$ is contained in a unique sector $S'\in V(\Gamma_{\leq i-1})$ and a unique sector $S^*\in V(\Gamma_i)$. Since both $S'$ and $S^*$ are supergraphs of $S$, the induction hypothesis implies that each contains a baseline; denote these baselines by $P_{S'}$ and $P_{S^*}$, respectively.

Towards proving the first statement  of the lemma, we distinguish two cases:

\noindent \textbf{(1)} $P_{S'}$ and $P_{S^*}$ are parallel. 
 Since a vertex of $S$ is reachable from both baselines by a straight path, either a vertex of $P_{S'}$ and a vertex of $P_{S^*}$ are reachable via straight paths 
 in the same direction from $S$, or via straight paths in opposite directions. In the former case, we obtain a baseline for $S$ by selecting the closer of $P_{S'}$ and $P_{S^*}$ and restricting it to the subpath whose vertices are reachable from the other baseline via a straight path. In the latter case, we similarly take one of the two baselines (say $P_{S'}$) and restrict it to the subpath whose vertices are reachable from the other baseline via a straight path. 
 
\noindent \textbf{(2)} $P_{S'}$ and $P_{S^*}$ are orthogonal. 

By Lemma~\ref{lem:sectortree}, the vertices of $S^*$ are exactly those reachable by straight paths $Q_1$, \dots, $Q_{|V(P_{S^*})|}$ starting from $P_{S^*}$ in a fixed direction. \\ Let $Z=\{v\in S~|~v \text{ is the first vertex encountered by some } Q_i, i\in [|V(P_{S^*})|]\}$. Every vertex of $Z$ is reachable by a straight path from some vertex of $P_S$, and $Z$ forms a path in $S$. Moreover, since every vertex of $S$ is reachable by a straight path from $P_{S^*}$, it is also reachable by a straight path from $Z$. Hence, $Z$ is a baseline for $S$. 

 It remains to establish the separation property. By the induction hypothesis, the existence of a baseline $P_j$ separating $S$ from terminal $p_j$ in $\Gamma_j$ holds for all $j\in [i]$, since for $j\in [i-1]$, we can invoke the induction hypothesis on the sector graph $\Gamma_{[i]-\{j\}}$ for the port sequence $(p_1,\dots, p_{j-1}, p_{j+1},\dots,p_i)$; this is correct, since the sectors are invariant under the order w.r.t.~which the ports are added.
 In all cases, the induction hypothesis guarantees the separability of a \emph{supergraph} of $S$ from $p_j$. Observe that the property of a set of vertices being separated from $p_j$ is closed under taking vertex-subsets. 

Finally, the fact that these baselines are parallel to a baseline for $S$ follows from the case distinction; indeed, in either case, we obtain that $S$ has a baseline that is parallel to baseline $P_{S^*}$, which is equal to $P_j$, if $p_j$ is the last added port. 
\end{proof}

Next, we formalize how ``adding'' a port can affect the structure of a sector; this will be useful later in our induction-based arguments.
\begin{definition}[Clean Sectors]
\label{def:clean}
For $j\in [\kappa]\setminus \{1\}$, we say that a sector $S \in V(\Gamma_{\leq j})$ is \emph{clean} (for $\Gamma_{\leq j}$) if $\Gamma_{\leq j}-S$ contains a connected component $\mathcal{C}_\text{main}$ such that the following properties hold:
\begin{enumerate}
\item all endpoints $p_1,\dots,p_j$ belong to sectors in $\mathcal{C}_\text{main}$, 
 \item every connected component of $\Gamma_{\leq j}-S$ other than $\mathcal{C}_\text{main}$ is a tree, and
\item $S$ has exactly one neighbor in each connected component of $\Gamma_{\leq j}-S$.
\end{enumerate}
\end{definition}

For completeness, we also define the notion of clean sectors for $\Gamma_1$ by simply declaring that every sector of $\Gamma_1$ is clean.

\begin{lemma}
\label{lem:sectorsplit}
Let $i\in [\kappa-1]$, let $S^*$ be a sector in $\Gamma_{\leq i}$ and let $\sect$ be the set of all sectors of $\Gamma_{\leq i+1}$ contained in $S^*$. Then at most six sectors in $\sect$ are not clean for $\Gamma_{\leq i+1}$.
 \end{lemma}

\begin{proof}
 If $S^*$ is a single vertex or a straight path, the lemma is trivial since $|\sect|\leq 3$ by Lemma~\ref{lem:pathsectors}. Hence, we may assume that $S^*$ has a baseline $P_{S^*}$. Without loss of generality, assume that $S^*$ lies in the $\uparrow$ direction from $P_{S^*}$---in particular, this implies that $P_{S^*}$ is horizontal (the other cases are symmetric). We next identify a special sector $S$ in $S^*$.

\begin{claim}
\label{claim:uniquemin}
There is a unique sector $S\in \sect$ with minimum bend distance to $p_{i+1}$.
\end{claim}

\begin{claimproof}
Recall that $\Gamma_{i+1}$ is a tree. If two distinct sectors $Z_1, Z_2 \in \sect$ had the same minimum bend distance to $p_{i+1}$, then by Lemma~\ref{lem:sectortree} neither has bend distance $0$ to $p_{i+1}$. Each would therefore have a neighbor with smaller bend distance to $p_{i+1}$ that lies outside of $S^*$, and adjacent to that neighbor they would have a baseline path that serves as a separator in $G$ (Claim~\ref{claim:treesep}). But then $G$ would contain a cycle $C$ whose internal part contains a boundary vertex, contradicting the fact that $G$ is a discretized polygon. Indeed, we can construct $C$ by following minimum-bend distance paths from $Z_1$ and $Z_2$ to $p_{i+1}$ until these intersect, and connecting these two paths inside $S^*$.
\end{claimproof}

The sector $S$ from Claim~\ref{claim:uniquemin} is the intersection of $S^*$ with some sector $Z$ in $\Gamma_{i+1}$ of bend distance $b$ to $p_{i+1}$.
To prove the lemma, we distinguish several cases based on the structure of $Z$; exhaustiveness follows from the characterization given in Lemma~\ref{lem:sectortree}.

\medskip

\noindent \textbf{Case 1:} $Z$ has a vertical baseline $P_Z$.

\begin{figure}[h]
    \centering
    \includegraphics[width= 0.7\linewidth,page=15]{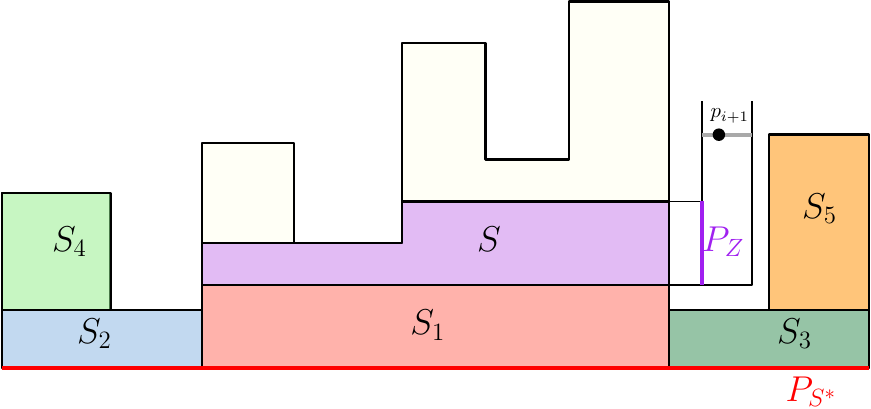}
    \caption{A schematic illustration for the most general configuration arising in Case 1.}
    \label{fig:non_clean_sectors}
\end{figure}
An illustration of the most general situation in this case is provided in Figure~\ref{fig:non_clean_sectors}.

Since $\pred(Z)$ does not intersect $S^*$, Lemma~\ref{lem:baseline} (and the construction underlying its proof) implies that the vertices of $S$ nearest to $P_Z$ form a baseline $P_S$ of $S$. Let $S^\uparrow\subseteq S$ and $S^\downarrow\subseteq S$ denote the paths orthogonal to $P_S$ farthest and closest from the baseline $P_{S^*}$ of $S^*$, respectively. 

All vertices immediately below $S^\downarrow$ have bend distance $b+1$ to $p_{i+1}$. Hence, either there is a unique sector $S_1 \in \sect$ of $\Gamma_{\leq i}$ adjacent to the vertices of $S^\downarrow$, or $S^\downarrow=P_{S^*}$. In the former subcase, notice that $S_1$ intersects $P_{S^*}$, and there is at most one sector $S_2$ containing the subpath $P_{S^*} - S_1$ in the $\leftarrow$ direction from $S_1$, and at most one sector $S_3$ which contains the subpath $P_{S^*} - S_1$ in the $\rightarrow$ direction from $S_1$; all vertices on those subpaths have bend distance $b+2$ to $p_{i+1}$. Since all vertices in $S^*$ are reachable via a vertical path from $P_{S^*}$, each remaining sector in $\sect\setminus \{S,S_1,S_2,S_3\}$ is then:

\begin{itemize}
\item either an $\uparrow$-neighbor of $S$ (adjacent to at least one vertex of $S^\uparrow$ and with bend distance $b+1$ to $p_{i+1}$), or
\item an $\uparrow$-neighbor of $S_2$ or $S_3$ (in which case it has bend distance $b+3$ to $p_{i+1}$). 
\end{itemize}

In particular, every sector in $\sect\setminus \{S,S_1,S_2, S_3\}$ is obtained by intersecting $S^*$ with a sector in $V(\Gamma_{i+1})$ whose baseline is horizontal. 
If one of $S_2$ and $S$ has an $\uparrow$-neighbor reachable via a vertical path from the leftmost vertex of $P_{S^*}$, denote it $S_4$; similarly, if either $S_3$ or $S$ has an $\uparrow$-neighbor reachable via a vertical path from the rightmost vertex of $P_{S^*}$, denote it $S_5$. If $\sect\setminus \{S,S_1,S_2, S_3, S_4, S_5\}=\emptyset$, the case is resolved. Otherwise, let $S_6\in \sect\setminus \{S,S_1,S_2, S_3, S_4, S_5\}$ be any remaining sector in $\sect$. Let $Z_6$ be the unique sector of $\Gamma_{i+1}$ containing $S_6$, and notice that the baseline $P_{Z_6}$ is horizontal (because $Z_6$ is adjacent to sectors $S$, $S_2$ or $S_3$, and all of these are contained in sectors of $\Gamma_{i+1}$ with a vertical baseline) and fully contained inside $S^*$ (because we excluded $S_4$ and $S_5$). 

We claim that $S_6$ is clean w.r.t.\ the connected component $\mathcal{C}_\text{main}$ of $\Gamma_{i+1}-S_6$ containing $S$. First, we argue that Property 1.\ from Definition~\ref{def:clean} holds; towards this, assume for a contradiction that $\Gamma_{i+1}-S_6$ contains a connected component $\mathcal{C}'\neq \mathcal{C}_\text{main}$ such that $\mathcal{C}'$ contains some port $p_\iota$. Note that by our choice of $S$ and the fact that $S\neq S_6$, $\iota\neq i+1$. Let $\alpha$ be the leftmost vertex on $P_{S^*}$, and $\beta$ be an arbitrary vertex of $P_{S^*}$ which lies on a vertical path intersecting $S_6$. Then $\bdist(\alpha,p_\iota)\geq \bdist(\beta,p_\iota)+1$ because every $\alpha$-$p_\iota$ path must cross $P_{S_6}$, 
as it is a separator per Lemma~\ref{lem:sectortree}. This would contradict the fact that $S^*$ is a sector of $\Gamma_i$, and thus we conclude that Property 1.\ holds.

Towards the remaining properties, let us consider a connected component $\mathcal{C}'\neq \mathcal{C}_\text{main}$ of $\Gamma_{i+1}-S_6$ and a sector $S_6'\in \mathcal{C}'$ that is a neighbor of $S_6$. Since no endpoint may lie in $\mathcal{C}'$ and the baseline $P_{S_6}$ 
is parallel to $P_{S^*}$ and fully contained in $S^*$, the bend distance of every vertex in $S'_6$ to \emph{every} port $p_{j}$, $j\in [i+1]$ is precisely one greater than the bend distance to the respective endpoint in $S_6$. By iteratively applying the same argument to the set of neighbors of $S'_6$, their neighbors and so forth, we obtain that every sector in $\mathcal{C}'$ is also a sector of $\Gamma_{\leq i}$ and of $\Gamma_{i+1}$. By the latter, we have that $\mathcal{C}'$ is a tree with a single neighbor $\pred(S'_6)=S_6$ in $\Gamma_{i+1}$. Thus, Properties 2.\ - 3.\ hold, completing the case.

\medskip

\noindent \textbf{Case 2:} $Z$ has a horizontal baseline $P_Z$.

Since $\pred(Z)$ does not intersect $S^*$, Lemma~\ref{lem:baseline} (and the construction underlying its proof) ensures that $S\cap P_{S^*}$ is non-empty and forms a baseline for $S$. Indeed, $S=S^*\cap Z$ whereas every vertex of $S^*$ admits a vertical path to $P_{S^*}$ and $Z$ consists of maximal vertical paths starting from a horizontal baseline $P_Z$ which either lies on the ``boundary'' of $S^*$ or outside of $S^*$. From this point, we can adapt the ideas of \textbf{Case 1} as follows. 

Let $S_2$ and $S_3$ be the unique $\leftarrow$- and $\rightarrow$-neighbors of $S$, respectively, and note that these two sectors contain all vertices with bend distance $b+1$ to $p_{i+1}$. All remaining sectors in $\sect$ are $\uparrow$-neighbors of $S_2$ and $S_3$ and have bend distance $b+2$ to $p_{i+1}$. Let $S_4$ be the leftmost $\uparrow$-neighbor of $S_2$ and $S_5$ the rightmost $\uparrow$-neighbor of $S_3$. If $\sect\setminus \{S,S_2, S_3, S_4, S_5\}=\emptyset$, then the case is resolved; otherwise, let $S_6\in \sect\setminus \{S,S_2, S_3, S_4, S_5\}$. To settle this case, it remains to establish that $S_6$ is clean when we set $\mathcal{C}_\text{main}$ to be the connected component of $\Gamma_{i+1}-S_6$ containing $S$; this argument is entirely analogous as in the last two paragraphs of \textbf{Case 1}.

\medskip

\noindent \textbf{Case 3:} $Z$ is a maximal vertical path.

Since $Z$ is maximal and $P_{S^*}$ is a horizontal baseline, $Z$ must intersect $P_{S^*}$. The rest of the argument follows analogously to \textbf{Case 2} where we set $P_Z$ to be the unique vertex in $V(P_{S^*})\cap V(Z)$.

\medskip

\noindent \textbf{Case 4:} $Z$ is a maximal horizontal path not fully contained in $S^*$.

This case is handled analogously to \textbf{Case 1} by taking $P_Z$ to be the unique vertex in $V(P_{S^*})\cap V(Z)$.

\medskip

\noindent \textbf{Case 5:} $Z$ is a maximal horizontal path that is fully contained in $S^*$.

This case can be handled by adapting the arguments for \textbf{Case 1}. Since $S=Z$ and all vertices immediately below $Z$ have bend distance $1$ to $p_{i+1}$, either there is a unique sector $S_1 \in \sect$ of $\Gamma_{\leq i}$ that lies below $Z$, or $Z=P_{S^*}$. In the former subcase, notice that $S_1$ intersects the baseline $P_{S^*}$, and there is at most one sector $S_2$ which contains the subpath $P_{S^*} - S_1$ in the $\leftarrow$ direction from $S_1$ and at most one sector $S_3$ which contains the subpath $P_{S^*} - S_1$ in the $\rightarrow$ direction from $S_1$; indeed, all vertices on those subpaths have bend distance $2$ to $p_{i+1}$. Since all vertices in $S^*$ are reachable via a vertical path from $P_{S^*}$, each remaining sector in $\sect\setminus \{S,S_1,S_2\}$ is now:

\begin{itemize}
\item either an $\uparrow$-neighbor of $S$ (i.e., is adjacent to at least one vertex of $Z$ and with bend distance $1$ to $p_{i+1}$), or
\item an $\uparrow$-neighbor of $S_2$ or $S_3$ (in which case it has bend distance $3$ to $p_{i+1}$). 
\end{itemize}

In particular, every sector in $\sect\setminus \{S,S_1,S_2, S_3\}$ is obtained by intersecting $S^*$ with a sector from $V(\Gamma_{i+1})$ whose baseline is horizontal. If either $S_2$ or $S$ has an $\uparrow$-neighbor that is reachable via a vertical path from the leftmost vertex of $P_{S^*}$, denote it $S_4$; similarly, if either $S_3$ or $S$ has an $\uparrow$-neighbor that is reachable via a vertical path from the rightmost vertex of $P_{S^*}$, denote it $S_5$. If $\sect\setminus \{S,S_1,S_2, S_3, S_4, S_5\}=\emptyset$, then the case is resolved; otherwise, let $S_6\in \sect\setminus \{S,S_1,S_2, S_3, S_4, S_5\}$. Let $Z_6$ be the unique sector of $\Gamma_{i+1}$ that contains $S_6$, and notice that the baseline $P_{Z_6}$ is horizontal (because $Z_6$ is adjacent to sectors $S$, $S_2$ or $S_3$, and all of these are contained in sectors of $\Gamma_{i+1}$ with a vertical baseline) and fully contained inside $S^*$ (because we excluded $S_4$ and $S_5$). 

To complete the proof of the lemma for this case, it remains to argue that $S_6$ is clean when we set $\mathcal{C}_\text{main}$ to be the connected component of $\Gamma_{i+1}-S_6$ containing $S$. This argument is entirely identical as that in the final two paragraphs in \textbf{Case 1}.
 \end{proof}

Lemma~\ref{lem:sectorsplit} guarantees that the addition of a single endpoint affects each existing sector only by subdividing it into a constant number of ``non-clean'' new sectors, along with a (potentially large) number of clean sectors. The properties of clean sectors ensure that they cannot appear in any large $2$-connected components, yielding the following theorem.

\begin{theorem}
\label{thm:sectortw}
The sector graph $\Gamma_{\leq \kappa}$ has treewidth at most $7^{\kappa-1}$.
\end{theorem}

\begin{proof}
We prove the theorem by induction on $i\in [\kappa]$. For the base case $i=1$, Lemma~\ref{lem:sectortree} implies that the treewidth of $\Gamma_1$ is at most $1$. Assume inductively that $\Gamma_{\leq i-1}$ has treewidth at most $7^{i-1}$,
witnessed by a tree decomposition $(T,\chi)$. We construct a tree decomposition $(T',\chi')$ of $\Gamma_{\leq i}$ as follows.

We first form the tuple $(T,\chi_0)$ from $(T,\chi)$ by replacing, in every bag, each sector $S'\in V(\Gamma_{\leq i-1})$ with the at most $6$ non-clean sectors of $\Gamma_{\leq i}$ contained in $S'$ as guaranteed by Lemma~\ref{lem:sectorsplit}. Note that for every node $t\in V(T)$, we have $|\chi_0(t)|\leq 6\cdot |\chi(t)|$. However, since clean sectors of $\Gamma_{\leq i}$ are not included in any of bag of $\chi_0(t)$ (for any $t\in V(T)$), the tuple $(T,\chi_0)$ is not yet a valid tree decomposition. 

To obtain a valid decomposition, we apply the following \emph{splitting} procedure exhaustively. Let $Q\in V(\Gamma_{\leq i})$ be a clean sector such that $Q\not \in \chi_0(t)$ for all $t\in V(T)$ (that is, $Q$ has not been incorporated into $\chi_0$ yet). Let $Q'$ be the unique sector of $\Gamma_{\leq i-1}$ containing $Q$. Recall that for each connected component $\mathcal{C}$ of $\Gamma_{\leq i}-Q$ such that $\mathcal{C}\neq \mathcal{C}_\text{main}$, $Q$ separated $\mathcal{C}$ from $\mathcal{C}_\text{main}$ and has exactly one neighbor in $\mathcal{C}$. Denote the components of $\Gamma_{\leq i}-Q$ by $\mathcal{C}_1=\mathcal{C}_\text{main}, \mathcal{C}_2,\dots, \mathcal{C}_p$, where $p$ is the total number of connected components of $\Gamma_{\leq i}-Q$. 

For each component $\mathcal{C}_j$ where $2\leq j\leq p$, we create a disjoint copy $T_k$ of the tree $T$ as it exists at the beginning of the current iteration of the splitting procedure. We modify $\chi_0$ to assign all sectors in $\mathcal{C}_j$ only to the corresponding bags in $T_j$---in particular, each sector in $\Gamma_{\leq i}$ is now in the $\chi_0$-image of nodes in precisely one of the created copies of $T$. Next, create a path $P$ of $p$
 entirely new nodes in $T$. We connect the 1st node of $P$ to an arbitrary node $c\in V(T_1)$ satisfying $Q\subseteq \chi(c)$, and for each $2\leq j\leq p$, we connect the $j$-th node to an arbitrary node whose bag contains
 the unique neighbor of $Q$ in $\mathcal{C}_j$. Finally, for each $j$-th node in $P$, set its bag contents to be precisely equal to the contents of its neighboring bag in $T_j$. Add $Q$ to every node in $P$.

This procedure preserves the interpolation property for all sectors which were added via the splitting procedure, as well as for the original non-clean sectors. In particular, the bags of each pair of adjacent vertices on the newly created paths only share $Q$ but are otherwise entirely disjoint. After exhaustively applying the splitting procedure is to all clean sectors, the tuple $(T,\chi_0)$ becomes a (potentially exponentially large) tree decomposition of $\Gamma_{\leq i}$ of width at most $7^i$, since each bag has size upper-bounded by $6$ times the maximum bag size of $(T',\chi')$ (which is $7^{i-1}$) plus one. 
\end{proof}

For the subsequent arguments, we also require an upper bound on the number of non-clean neighbors of each sector.

\begin{lemma}
\label{lem:fewnonclean}
Each sector of $\Gamma_{\leq \kappa}=\Gamma$ has at most eight non-clean neighbors.
\end{lemma}

\begin{proof}
Let $S$ be a sector of $\Gamma$ that is not a single vertex (the statement is trivial for single-vertex sectors). We say that a sector $S'\in N_{\Gamma}(S)$ is a \emph{total neighbor} of $S$ if there exists a direction $d\in \directions$ such that, for each vertex $v\in S$, the straight path from $v$ in direction $d$ will encounter a vertex $w\in S'$ and this $w$ is the first vertex encountered by the path outside of $S$. In other words, $S'$ is a total neighbor of $S$ if it can be reached from every vertex of $S$ by moving in the same direction. Since there are $4$ directions, $S$ can have at most $4$ total neighbors. 

If $S$ has no total neighbor in direction $d$, then it has a sequence of at least two neighbors reachable by maximal straight paths in direction $d$, ordered along the axis orthogonal to $d$. We call such neighbors \emph{partial neighbors}, and the first and last partial neighbors in direction $d$ are called \emph{corner neighbors}. Consequently, $S$ has at most $8$ corner and total neighbors overall.

To prove the lemma, it suffices to show that every neighbor of $S$ that is neither a corner neighbor nor a total neighbor must be clean.

Let $S'$ be a non-corner non-total $\downarrow$-neighbor of $S$; the arguments for the other $3$ directions are symmetric. Let $v_\texttt{left}$ and $v_\texttt{right}$ denote the $\leftarrow$-most and $\rightarrow$-most vertices in $S'$ that have an $\uparrow$-neighbor in $S$, respectively. 

Consider the situation where $v_\texttt{left}$ has a $\leftarrow$-neighbor $w$ in $G$; by definition of $v_\texttt{left}$, we have $w\not \in S'$. In particular, there exists a port $p_i$ such that $\bdist(v_\texttt{left},p_i)\neq \bdist(w,p_i)$. Now, let us consider the sectors $S_v$, $S_w$ in $\Gamma_{i}$ that contain $v_\texttt{left}$ and $w$, respectively. By Lemma~\ref{lem:sectortree}, either $S_v=\pred(S_w)$ or $S_w=\pred(S_v)$. In either case, $S_v$ or $S_w$ has a baseline $P$, and since $\Gamma_i$ is a tree, $P$ is a straight line which passes through the sector $S\in V(\Gamma)$. Crucially, $S'$ is not a corner sector and so $P$ separates at least two vertices of $S$---a contradiction with the fact that all vertices in $S$ have the same bend distance to $p_i$.

Thus, $v_\texttt{left}$ cannot have a $\leftarrow$-neighbor in $G$, and for the same reason $v_\texttt{right}$ cannot have a $\rightarrow$-neighbor in $G$ either. Note that if $S'$ has a single neighbor $S$ in $\Gamma$, then it is trivially clean and we are done. The only possibility that remains is that $S'$ is a separator in $\Gamma$. 

We now argue that the only connected component of $\Gamma-S'$ containing a port is the connected component containing $S$. Indeed, consider a hypothetical port $p_\circ$ that lies in a different connected component of $\Gamma$ than $S$. There exist two vertices $v_1, v_2\in S$ such that $v_1$ can reach $S'$ with a straight path, but $v_2$ cannot, and every path to $p_i$ must pass through $S'$. This yields $\bdist(v_2,p_i)=\bdist(v_1,p_i)+1$---contradicting the fact that $S$ is a sector. 

Having established that $S'$ is a separator and the only connected component $\mathcal{C}_\text{main}$ containing any port is the one of $S$, we proceed towards establishing Properties 2.\ and 3.\ of Definition~\ref{def:clean}. This can be achieved via the same arguments as those used in the proof of Lemma~\ref{lem:sectorsplit} (Case 1). In particular, let us consider a connected component $\mathcal{C}'\neq \mathcal{C}_\text{main}$ of $\Gamma-S'$ and a sector $S^*\in \mathcal{C}'$ that is a neighbor of $S'$. Then the bend distance of every vertex in $S^*$ to \emph{every} port $p_{j}$, $j\in [\kappa]$, is precisely one greater than the bend distance to the respective endpoint in $S'$. By iteratively applying the same argument to the set of neighbors of $S^*$, their neighbors and so forth, we obtain that every sector in $\mathcal{C}'$ is adjacent to a single sector with smaller bend distances, and in particular that $\mathcal{C}'$ is a tree with a single neighbor $S'$ in $\Gamma$. Thus, Properties 2.\ and 3.\ hold and $S'$ is clean, as desired.
\end{proof}

\section{Reducing the Sectors}
\label{sec:reduction}
Let $\I=(G, \R, k, \ell)$ be an instance of \dcmp{}. Throughout this section, we assume that if $I$ is a \YES-instance, then it admits a solution satisfying the conditions in Assumption~\ref{ass:canonical}. In particular, every robot in $\R$ has travel slack at most $\lambda =\Oh(k^5)$, and there are at most $\mu = k^{\Oh(k^2)}$ turns and waiting points within any rectangular subgrid of $G$. We further assume that $\mu > \lambda > k+1$.   
\begin{observation}
\label{obs:secstructure}
A sector $S \in \Gamma$ belongs to one of the following categories:
\begin{itemize}
\item[(i)] $S$ is a single vertex.
\item[(ii)] $S$ is straight path in $G$.
\item[(iii)] $S$ has a single baseline, in which case we refer to $S$ as a \emph{histogram}.
\item[(iv)] $S$ has exactly two orthogonal baselines, in which case we refer to $S$ as a \emph{staircase}.
\item[(v)] Otherwise, $S$ is a rectangle; that is, the subgraph of $G$ induced by the vertices in $S$ is a full rectangular grid.  

\end{itemize}
\end{observation}

 \begin{figure}
     \centering
     \includegraphics[page=10,width = 0.7\linewidth]{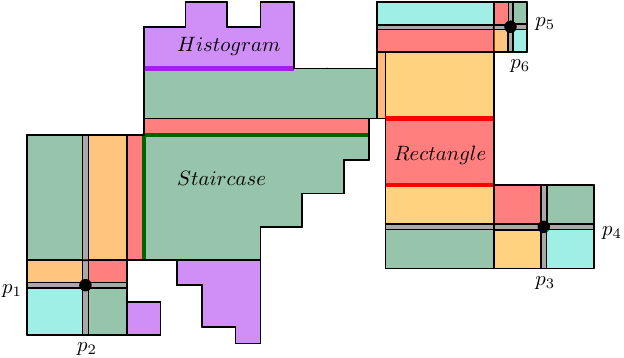}\includegraphics[page=12, width=0.3\linewidth]{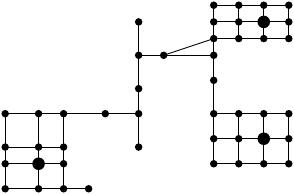}
     \caption{An illustration of sectors in a discretized polygon w.r.t. three terminals defining six ports $p_1, \ldots, p_6$. Three of the sectors have highlighted their baselines, showcasing the three non-trivial types of a sector. On the right is the sector graph $\Gamma$ defined by the sectors.}
     \label{fig:placeholder}
 \end{figure}
Note that case (v) above follows from the fact that if a sector has two parallel baselines, then all points between these two baselines (including the points on the baselines) must be sector points. 

\begin{lemma}\label{lem:shortest_path_through_baseline}
    Let $S\in \Gamma$ be a sector that is neither a singleton nor a straight path in $G$, let $u$ be a vertex in $S$, and let $v$ be a point in \textsf{Terminals}. Then there exists a shortest path from $v$ to $u$ that intersects a baseline of $S$.
\end{lemma}
\begin{proof}

Fix a direction $d\in \bidirections$ and consider the port $p_i=(v,d)$. The sector $S\in \Gamma$ is a subset of a sector $S_i\in \Gamma_i$. Since $S$ is not a straight path in $G$,  Lemma~\ref{lem:sectortree} implies that $S_i$ has a $b$-baseline $P_{S_i}$, for some $b\in \directions$, which forms a cut between $S_i$ and $v$. By Lemma~\ref{lem:baseline}, the sector $S$ has a $b$-baseline $P_S^b$. 

Since $P_{S_i}$ is a baseline, by definition, $u$ is reachable by a straight path that starts at a vertex $x\in V(P_{S_i})$ and follows direction $b$. Moreover, since $S$ is a subset of $S_i$ and $u$ is reachable from $P_S^b$ by following a straight path from $P_S^b$ in direction $b$, it follows that the straight path from $P_{S_i}$ to $u$ intersects $P_S^b$. 

 Consider a shortest path $P_{vu}$ from $v$ to $u$. Since $P_{S_i}$ is a separator between $v$ and $S_i$, the path $P_{vu}$ must intersect $P_{S_i}$.
Let $y\in V(P_{S_i})$ be the first vertex on $P_{vu}$ that also lies on $P_{S_i}$. Since $P_{S_i}$ is a straight-line orthogonal to direction $b$ and $G$ is a subgrid, the path obtained by following $P_{S_i}$ from $y$ to $x$ and then following direction $b$ from $x$ to $u$ is a shortest path from $y$ to $u$. Hence, the path $P'_{vu}$ that follows $P_{vu}$ from $v$ to $y$, then follows $P_{S_i}$ from $y$ to $x$, and finally follows direction $b$ from $x$ to $u$ is a shortest path from $v$ to $u$ that intersects the baseline $P_S^b$ of $S$. 
\end{proof}

\begin{lemma}
\label{lem:histogramreduction}
Let $S\in \Gamma$ be a histogram with baseline $P_S$. Let $v$ be a vertex in $S$ whose distance from $P_S$ is $\delta$. Then, for any two endpoints $p$ and $q$ in $\textsf{Terminals}$, any $p$--$q$ path containing $v$ has length at least $\dist(p, q) + 2 \delta$.
\end{lemma}

\begin{proof}
By Lemma~\ref{lem:shortest_path_through_baseline}, there exists a shortest path $P$ from $p$ to $v$ that intersects $P_S$, and a shortest path $Q$ from $q$ to $v$ that intersects $P_S$. Clearly, any path from $p$ to $q$ that passes through $v$ has length at least that of the concatenation of $P$ and the reverse of $Q$. Therefore, it suffices to show that the path obtained by concatenating $P$ with the reverse of $Q$, denoted $PQ^{rev}$, has length at least $\dist(p, q) + 2\delta$. 

Let $x$ be the vertex at which $P$ first intersects $P_S$, let $y$ the vertex at which $Q$ first intersects $P_S$, and let $z$ be the projection of $v$ onto the baseline $P_S$ (that is, $z$ is the vertex of $P_S$ closest to $v$). Note that $\dist(z, v)= \delta$. Consider a $p$--$q$ path $P'$ that follows $P$ from $p$ to $x$, then follows $P_S$ from $x$ to $y$, and finally follows the reverse of $Q$ from $y$ to $q$. Along $P'$, the path from $x$ to $y$ is a shortest path, and hence $\dist(x,y)\le \dist(x,z)+\dist(z,y)$. 

In contrast, the path $PQ^{rev}$ goes from $x$ to $v$ and then from $v$ to $y$ and $\dist(x,v) = \dist(x,z)+\dist(z,v)$. Since $G$ is a grid, the shortest paths from $x$ to $z$ and from $z$ to $v$ are both straight paths. Therefore, $\dist(v,y) = \dist(v,z)+\dist(z,y)$. Hence  $\dist(x,v)+ \dist(v,z) \ge \dist(x,y) + 2\delta$. It follows that the length of $P'$ is shorter than that of $PQ^{rev}$ by at least $2\delta$, and the proof is complete.
\end{proof}

This reduction follows from Lemma~\ref{lem:histogramreduction} and Assumption~\ref{ass:canonical}:

\begin{reduction}
\label{red:hist}
 Let $\I$ be an instance of \dcmp, and let $\Gamma$ be the sector graph of $\I$. Let $S$ be a histogram sector in $\Gamma$ with baseline $P_S$. Let $\I'$ be the instance obtained from $\I$ by removing every vertex in $S$ whose distance to $P_S$ is greater than $\lambda/2$; see Figure~\ref{fig:red_hist}. Then $\I'$ is equivalent to $\I$. 
\end{reduction}

\begin{figure}
    \centering
    \includegraphics[page=1]{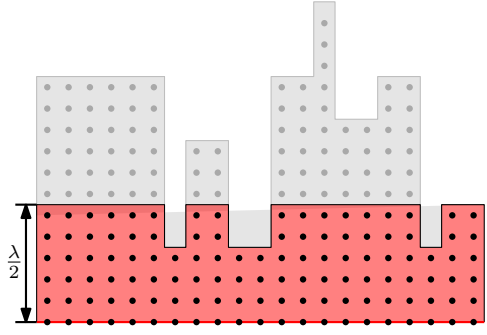}
    \caption{Illustration of Reduction Rule~\ref{red:hist}. The gray vertices are removed from the graph by the reduction rule.}
    \label{fig:red_hist}
\end{figure}
Furthermore, we can apply analogous arguments to reduce the clean sectors and the components defined by them, yielding the following reduction rule.

\begin{reduction}
\label{red:cleansectors}
    Let $\I$ be an instance of \dcmp{}, and let $\Gamma$ be the sector graph of $\I$. Let $S$ be a clean sector in $\Gamma$ with baseline $P_S$, and let $\mathcal{C}_{main}$ be the connected component of $\Gamma-S$ that contains all endpoints $p_1, \ldots, p_{\kappa}$. Let $\I'$ be the instance obtained from $\I$ by removing every vertex in any sector $S'\in \Gamma\setminus \mathcal{C}_{main}$ that is at distance at least $\lambda/2$ from $P_S$. Then $\I'$ is equivalent to $\I$. 
\end{reduction}

\begin{proof}
    Since all ports are in sectors in $\mathcal{C}_{main}$, and $S$ has only one neighbor (by definition) in $\mathcal{C}_{main}$, it follows that $S$ has a unique baseline $P_S$ which serves as a separator between the vertices in sectors in $\mathcal{C}_{main}$ and those outside of it. The proof now follows from Assumption~\ref{ass:canonical}. 
\end{proof}

Let $S\in \Gamma$ be a staircase sector with two baselines $P^{1}_S$ and $P^2_S$, where $P_S^1$ is a $b_1$-baseline for $b_1\in \directions$ and $P_S^2$ is a $b_2$-baseline for $b_2\in \directions$, with $b_1$ and $b_2$ orthogonal. 
     Let $v$ be the intersection vertex of $P^1_S$ and $P^2_S$, and let \emph{the staircase base} $Q_S$ be the set of precisely all vertices that are reachable from $v$ via a straight path that starts at $v$ and follows either direction $b_1$ or direction $b_2$. Finally, let $C_S$ be the connected component of $G-Q_S$ containing $S$.

\begin{lemma}
    \label{lem:staircaseseparation}
    Let $S\in \Gamma$ be a staircase sector, and let $Q_S$ be its base. Then, the connected component $C_S$ of $G-Q_S$ that contains the vertices of $S$ does not contain any endpoint $u\in \textsf{Terminals}$. 
\end{lemma}

\begin{proof}
    Note that $Q_S$ is basically a union of two straight lines that go from the boundary to $v$. Hence, $Q_S$ forms a cut in $G$. 
    
    Let $u\in \textsf{Terminals}$, and consider port $p_i=(u,d)$ for some $d\in \bidirections$. The sector $S\in \Gamma$ is a subset of a sector $S_i\in \Gamma_i$. Since $S$ is not a straight path in $G$, by Lemma~\ref{lem:sectortree}, $S_i$ has a $b$-baseline $P_{S_i}$ for some $b\in \directions$. By Lemma~\ref{lem:baseline}, it follows that $b$ is either $b_1$ or $b_2$. Moreover, $P_{S_i}$ is a cut between $S_i$ and $u$ by Lemma~\ref{lem:sectortree}. It is then clear that $Q_S$ separates $u$ from $C_S$.      
\end{proof}

\begin{lemma}
\label{lem:staircasered}
    Let $\I=(G,\R,k,\ell)$ be an instance of \dcmp{}. Let $S\in \Gamma$ be a staircase sector, and let $Q_S$ be its base. Define $f(k)=(k+1)(\lambda+1)$. Let $G'$ be the graph obtained from $G$ by removing every vertex in $C_S$ whose distance from $Q_S$ exceeds $f(k)$. 
    
    If $\I$ is a \YES-instance, then there exists a solution to $\I$ that is confined to $G'$. 
\end{lemma}
\begin{figure}[h]
    \centering
    \includegraphics[ page=7]{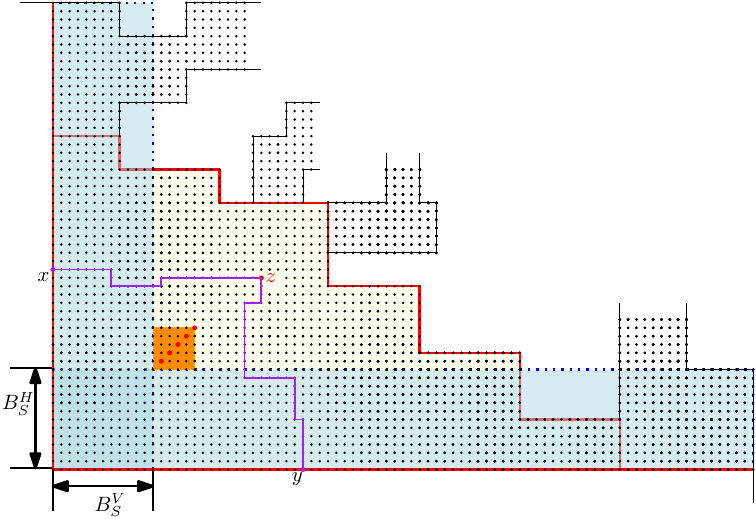}
    \caption{An illustration of a robot's path that reaches a point $z$ at distance at least $\lambda$ from both buffers along the portion of its route that intersects staircase sector $S$ between vertices $x$ and $y$. Since there exists a shortest path from $x$ to $y$ that follows the staircase base $Q_S$, it follows that along its route between $x$ and $z$, the robot is never more than $\frac{\lambda}{2}$ grid lines below $z$, and along its route between $y$ and $z$ it is never more than $\frac{\lambda}{2}$ grid line to the left of $z$. Consequently, the segment of the route between $x$ and $y$, together with $Q_S$, bounds a region that contains a square with side length $\frac{\lambda}{2} > k+1$ lying outside the buffers $B^H_S$ and $B^V_S$. }
    \label{fig:existence_of_Omega}
\end{figure}

\begin{figure}[h]
    \centering
    \includegraphics[ page=6]{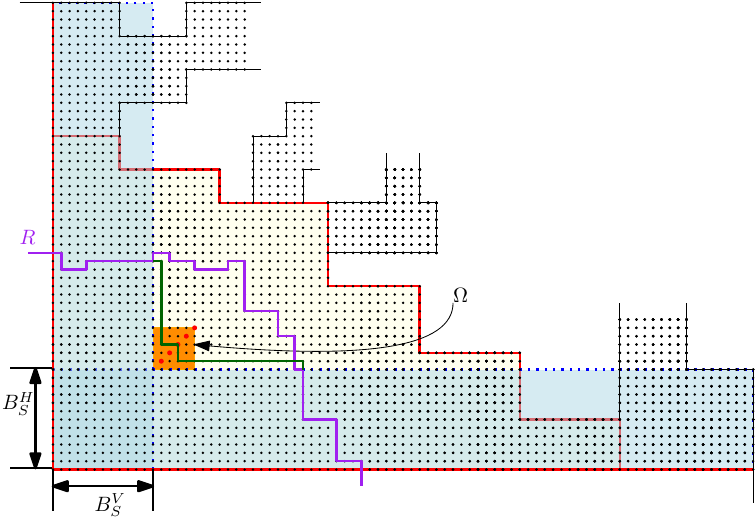}
    \caption{A new route for a robot $R$ that exits the buffers and visits a vertex at distance at least $\lambda$ from both buffers; the modified portion of the route is shown in green. Observe that such a route intersects the two orthogonal inner buffer lines of $S$ at a distance at least $\frac{\lambda}{2}$ from their intersection, and therefore lies outside $\Omega$. As soon as the robot intersects an inner buffer line, we guide it to its unique vertex in $\Omega$. Consequently, no two robots ever move along these lines at the same time.}
    \label{fig:staricase_rerouting}
\end{figure}
\begin{proof}
     Suppose that $\I$ has a solution $\SSS$. If $\SSS$ is already confined to $G'$, we are done. Otherwise, we proceed as follows.  
     
     We define a \emph{horizontal buffer} $B_{S}^{H}$ and a \emph{vertical buffer} $B_{S}^{V}$ of $C_S$. We first let $B_{S}^{H}$ be the set of horizontal lines in $C_S$ whose distance from the extension of the horizontal baseline of $S$ is at most $\lambda$. Similarly, $B_{S}^{V}$ is the set of vertical lines in $C_S$ within distance $\lambda$ from the extension of the vertical baseline of $S$. 
     If there exists a robot $R$ such that, in $\SSS$, $R$ leaves both buffers, but does not reach a point that is at distance at least $\lambda$ from the inner buffer lines in both the horizontal and the vertical directions, then we increase the width of both buffers by $\lambda$. Observe that, after increasing the buffer widths, the movement of $R$ is confined to the two buffers. 
     See Figure~\ref{fig:existence_of_Omega} for an illustration.     
     
     The process can repeat at most $k$ times (once per robot). Consequently, we end up with two buffers, each of width at most $\lambda\cdot (k+1)$. Moreover, every robot that leaves these buffers reaches a point that is at distance at least $\lambda$ from both buffer boundaries. Let $\Omega$ denote the square of dimension $k+1$ whose sides coincide with the two orthogonal inner buffer lines of $S$. 

    We show that if any robot leaves the buffers, then $\Omega$ lies entirely within sector $S$; see Figure~\ref{fig:existence_of_Omega}. Consider a robot $R$ that intersects the vertical baseline of $S$ at a vertex $x$ and the horizontal baseline of $S$ at a vertex $y$. Suppose that, along its subpath $P_{xy}$ between $x$ and $y$, the robot $R$ passes through a vertex $z$ that is at distance at least $\lambda$ from both buffers. 
    
    Observe that there exists a shortest path $P'_{xy}$ from $x$ to $y$ that first follows the vertical baseline line of $S$ from $x$ to the corner vertex $v$ of $Q_S$, and then continues from $v$ to $y$ along the horizontal baseline of $S$. It follows that, along $P_{xy}$, the robot $R$ can move at most $\frac{\lambda}{2}$ steps in a direction opposite to either of the two directions followed on $P'_{xy}$. Since $z$ is at distance at least $\lambda$ from both buffers, when $R$ moves between $x$ and $z$ along $P_{xy}$, it does not cross any horizontal line within distance $\frac{\lambda}{2}$ of $B^H_S$, and when it moves between $y$ and $z$, it does not cross any vertical line within distance $\frac{\lambda}{2}$ from $B^V_S$. 
    
    Consequently, the paths $P_{xy}$ and $P'_{xy}$ together form a closed curve that lies entirely within $C_S$ and fully contains the square $\Omega$, since $\lambda > k+1$. As $G$ contains no holes, it follows that all vertices of $\Omega$ belong to $G$.
     
     We now show how to use the two buffers and $\Omega$ to modify the solution $\SSS$ to $\I$ so that it is confined to $G'$; see Figure~\ref{fig:staricase_rerouting}. 
     
    We map the robots in $\R$ that exit the buffers bijectively to diagonal points of $\Omega$. For each robot $R$ that crosses an inner buffer line, we freeze the schedule at the moment it reaches the inner line of a buffer, and route $R$, using the first grid line outside the buffer that it crosses, to its uniquely assigned point in $\Omega$. The robot remains at that assigned point until the last time it intersects the other inner line of a buffer in the original schedule. At that moment, we freeze the schedule again and route $R$ to the point on the inner line where it is located at that moment in the original schedule. The rerouting of $R$ is illustrated by the green path in Figure~\ref{fig:staricase_rerouting}. 
\end{proof}

The following reduction rule follows immediately from the previous lemma.

\begin{figure}
    \centering
    \includegraphics[width=\linewidth,page=2]{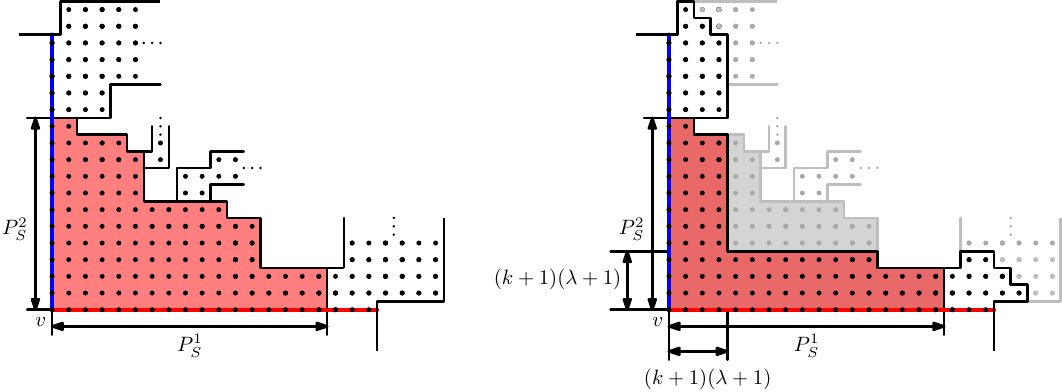}
    \caption{Illustration of Reduction Rule~\ref{red:stair}. The red region highlights a staircase sector $S$. The blue and red lines, along the left and bottom, respectively, are extensions of $S$'s baselines and together form the staircase base $Q_S$. The vertices outside the red region belong to $C_S\setminus S$. Gray vertices in the left figure indicate the vertices removed by the application of the reduction rule. }
    \label{fig:red_staircase}
\end{figure}
\begin{reduction}
\label{red:stair}
 Let $\I$ be an instance of \dcmp{}, and let $\Gamma$ be the sector graph of $\I$. Let $S\in \Gamma$ be a staircase sector with base $Q_S$. Let $\I'$ be the instance obtained from $\I$ by removing every vertex in the connected component $C_S$ of $G-Q_S$ that contains the vertices of $S$
 whose distance from $Q_S$ exceeds $(k+1)(\lambda+1)$; see Figure~\ref{fig:red_staircase}. Then $\I'$ is equivalent to $\I$.
\end{reduction}

Let $S$ be a rectangle sector with dimensions $a\times b$, and let $\SSS$ be a solution. We say that a row (or column) $i$ is \emph{safe} for $\SSS$ if no turn or waiting step of any robot $R$ occurs on the $i$-th row (or column) of $S$. A \emph{$\gamma$-safe row (or column) sequence} is a set of $\gamma$ consecutive safe rows (or columns).

We say that a schedule $\SSS$ is \emph{semi-canonical} if it satisfies the following condition. For every rectangle sector $S$ of dimensions $a\times b$, where $\min\{a,b\}\geq (k+1)(\mu^2+\mu)$, there exists a rectangle $S' \subseteq S$ obtained by removing at most $q(k)$ rows and columns from the top, bottom, left, and right, where $q(k)=(k+1)\mu^2+\lambda$ is a function independent of $S$, such that $S'$ satisfies the following properties:

\begin{enumerate}
\item For each robot $R$ that enters $S'$ for the first time in $\SSS$ from a direction $d\in\directions$, $R$ exits $S'$ for the last time in $\SSS$ in the same direction $d$.
\item The bottom-most $(k+1)\mu-\lambda$ rows and the left-most $(k+1)\mu-\lambda$ columns of $S'$, each form a safe sequence.
\end{enumerate}

We call a subgraph $S'$ satisfying the properties above a \emph{first frame} of $S$.

\begin{lemma}\label{lem:semi-canonical_solution}
Every \YES-instance $\I=(G, \R, k, \ell)$ of \dcmp admits a semi-canonical solution.
\end{lemma}
\begin{figure}
    \centering
    \includegraphics[page=5]{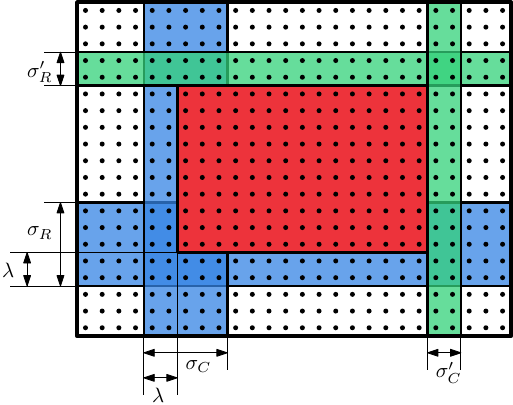}
    \caption{Illustration of the proof of Lemma~\ref{lem:semi-canonical_solution}. Here, $\sigma_C$ and $\sigma_R$, denote the left-most and the bottom-most $((k+1)\mu)$-safe sequences, respectively, while $\sigma'_C$ and $\sigma'_R$, denote the right-most and the top-most $\lambda$-safe sequences. The first frame corresponds to the red rectangle in the middle. Note that it also excludes the first $\lambda$ rows of $\sigma_R$ and first $\lambda$ columns of $\sigma_C$.}
    \label{fig:first_frame}
\end{figure}

\begin{figure}
    \centering
    \includegraphics[page=8]{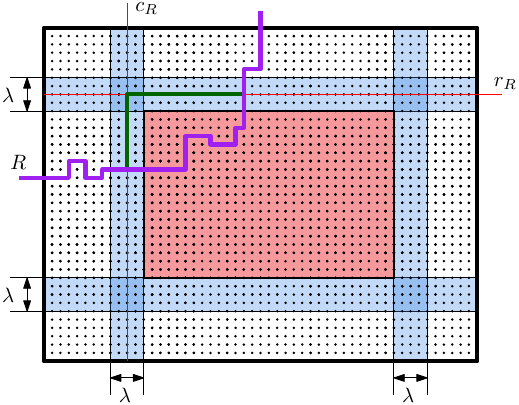}
    \caption{Illustration of a new route for a robot entering the first frame from the left and exiting at the top. The original route of $R$ is shown in purple, while the modified portion of the route---following column $c_R$ and then row $r_R$---is shown in green. Note that $c_R$ and $r_R$ were  free, as a column and a row in $S$ before this modification, and after all changes they are used exclusively by robot $R$.}
    \label{fig:first_frame_route_modification}
\end{figure}
\begin{proof}
Let $S$ be a rectangle sector with dimensions $a\times b$, where $\min\{a,b\}\geq (k+1)(\mu^2+\mu)$, and let $\SSS$ be a schedule satisfying Assumption~\ref{ass:canonical}. We now describe how to obtain the first frame of $S$ (see Figure~\ref{fig:first_frame}). 

Partition the columns of $S$ into column sequences, each with $(k+1)\cdot \mu$ many consecutive columns. Since $S$ has at least $(k+1)(\mu^2+\mu)$ many columns, the pigeon-hole principle guarantees the existence of a column-sequence $\sigma_C$ of $(k+1)\cdot \mu$ many consecutive free columns. By an analogous argument, there exists a row-sequence $\sigma_R$ with $(k+1)\cdot \mu$ many consecutive free rows that contains no turns or waiting points, and such that at most $(k+1)\mu^2$ rows in $S$ lie below $\sigma_R$. 

Similarly, we can find sequences $\sigma_C'$ and $\sigma_R'$ consisting of $\lambda$ many consecutive columns and rows, respectively, which contain no turns or waiting points and lie at most $(\lambda+1)\mu$ away from the right and top boundaries of $S$. We define $S'$ to be the subgraph of $S$ obtained by excluding:
\begin{itemize}
    \item all columns to the left of $\sigma_C$, 
    \item the $\lambda$ left-most columns of $\sigma_C$, 
    \item the columns of $\sigma_C'$ and all columns to the right of it, 
    \item all rows below $\sigma_R$,
    \item the $\lambda$ bottom-most of $\sigma_R$, and 
    \item the rows of $\sigma_R'$ and all rows above it. 
\end{itemize}

Now, consider a robot $R$ entering $S'$ from a direction $d\in \directions$. Since $S'$ is surrounded by at least $\lambda$ rows and columns on each side, in which $R$ does not turn or stop, and since $R$ has travel slack at most $\lambda$, $R$ cannot exit $S'$ in the direction opposite to $d$ (i.e., by crossing the same side of $S')$. If $R$ exits $S'$ in $\SSS$ in direction $d$, it already satisfies the conditions of a semi-canonical solution with respect to the first frame $S'$. 

Suppose instead that $R$ exits $S'$ in the direction $d'$ orthogonal to $d$ (see Figure~\ref{fig:first_frame_route_modification}). Without loss of generality, assume $R$ enters from the left and exits from the top; all other cases are analogous. Let $R$ enter in row $r$ and exit in column $c$. Since $c$ lies within $S'$, it is between $\sigma_C$ and $\sigma_C'$, as no robot uses rows in $\sigma_C$ and in $\sigma_C'$. Similarly, $r$ is between $\sigma_R$ and $\sigma_R'$. 

Moreover, $R$ crossed $\sigma_C$ entirely along row $r$ without waiting or turning, and crosses $\sigma_R'$ along column $c$, again without turning or waiting. We now modify $R$'s route as follows. 

We assign $R$ a column $c_R$ of $\sigma_C$, at a distance between $k$ and $2k$ from the left-most column of $\sigma_C$, and a row $r_R$ of $\sigma_R'$, at a distance between $k$ and $2k$ from the top-most row of $\sigma_R'$. Note that $c_R$ and $r_R$ are unique to $R$. When $R$ reaches $c_R$ along row $r$ in $\SSS$, we pause the schedule and move $R$ along $c_R$ to the vertex $v_R$ at the intersection of $c_R$ and $r_R$. This intersection is free, as it lies on both a free row and a free column. 

if a robot $R'$ occupies $c_R$ while moving $R$ to $v_R$ along $c_R$, $R'$ at this point must be moving perpendicular to $c_R$. The $k$ columns on either side of $c_R$ must be free, so we can shift $R'$ (and any of the at most $k$ robots moving with it on the same row) by one move in the direction of its motion to clear $c_R$ for $R$. 

After $R$ reaches $v_R$, it waits there until the timestep in $\SSS$ when it would reach row $r_R$ after leaving $S'$ via column $c$. At this timestep, we again pause $\SSS$ and move $R$ along $r_R$ to its position in column $c$, shifting any obstructing robots on its route as needed.
\end{proof}

Let the \emph{second frame} $S^*$ of $S$ be the subgraph of $S$ obtained by removing the top, bottom, left, and right $ (k+1)(\mu^2+\mu)$ rows and columns. 
We say that a schedule $\SSS$ is \emph{canonical} if 
 for each $a\times b$ rectangle sector $S$, where $\min\{a,b\}\geq (k+1)(\mu^2+\mu)$, 
each robot $R$ enters the second frame $S^*$ of $S$ only once, and never turns or waits in $S^*$.
 
\begin{lemma}\label{lem:canonical_solution}
Every \YES-instance $\I=(G, \R, k, \lambda)$ of \dcmp admits a canonical solution.
\end{lemma}
\begin{figure}
    \centering
    \includegraphics[page=9,width=\linewidth]{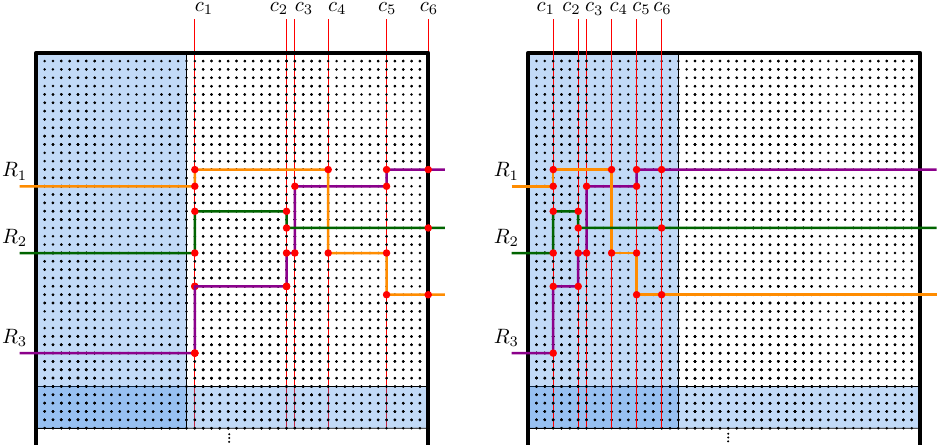}
    \caption{Illustration of proof of Lemma \ref{lem:canonical_solution}. Three robots traverse the first frame $S'$ of a rectangle sector $S$ between its left and right boundaries. Their turn and wait points lie on columns in $\CCC=\{c_1, c_2,\ldots, c_6\}$, which are outside the $k\cdot\mu$ leftmost columns that are free. We map the columns in $\CCC$ to the $k\mu$ leftmost columns, preserving pairwise distances up to $k$. Consequently, at most $k$ robots ($k-1$ in the middle and one at the end vertex) can occupy the same horizontal segment between two consecutive ``important'' columns. Note that columns $c_2$ and $c_3$ are at distance one before and after the shift, whereas all other consecutive column pairs are at distance $k=3$, as their original distance exceeds $3$.      }
    \label{fig:shifting_turns}
\end{figure}
\begin{proof}
Let $\SSS$ be the semi-canonical solution guaranteed by Lemma~\ref{lem:semi-canonical_solution}. 
Let $S$ be an $a\times b$ rectangle sector, and let $S'$ denote the first frame of $S$. Observe that the left-most $(k+1)\cdot \mu - \lambda \ge k\cdot \mu$ columns and bottom-most $(k+1)\cdot \mu - \lambda \ge k\cdot \mu$ rows of $S'$ are free. 

We define the second frame $S^*$ as the subgraph of $S'$ obtained by removing the left-most $k\cdot \mu$ columns and bottom-most $k\cdot \mu$ rows of $S'$. The main idea of the proof is to shift all waiting and turning vertices of robots that enter or leave $S'$ in a direction $d\in \{\leftarrow, \rightarrow\}$ into the left-most $k\cdot\mu$ columns of $S'$, and all the waiting and turning vertices of robots that enter/leave in a direction in $\{\uparrow, \downarrow\}$ to the bottom-most $k\cdot\mu$ rows of $S'$ (see Figure~\ref{fig:shifting_turns}).  This guarantees that, in $S^*$, each robot moves continuously along a straight path without making any turns.

Let $\mathcal{C} = \{c_1, c_2, \ldots, c_q\}$ be the set of all columns containing a turning or a waiting vertex of any robot that enters or leaves $S'$ in a direction in $\{\leftarrow, \rightarrow\}$. Assume that $c_i$ lies to the left of $c_{i+1}$, for all $i\in [q-1]$. Note that $q \le \mu$. We map the columns in $\mathcal{C}$ bijectively to the first $k\cdot \mu$ columns of $S'$ as follows. We map $c_1$ to the $k$-th leftmost column of $S'$. Inductively, given the image of $c_i$, we map $c_{i+1}$ to the column that is to the right of the image of $c_i$, at a distance equal to the minimum of $k$ and the distance between $c_i$ and $c_{i+1}$.

Now consider a robot $R$ that enters or leaves $S'$ in a direction $d\in \{\leftarrow, \rightarrow\}$. We modify the route of $R$, so that, whenever $R$ occupies a column $c_i$ in $\SSS$, it occupies the image of $c_i$ in the modified schedule $\SSS'$. We consider three cases: when $R$ enters $S'$, when $R$ leaves $S'$, and when $R$ moves between two columns in $\mathcal{C}$.  

When $R$ enters $S'$, we freeze the execution of $\SSS$ and move $R$ along its entry row as close as possible to the image of $c_1$ (or $c_q$) without occupying that column. There may exist another robot $R'$ moving along the same row as $R$ in the original schedule. However, since no robot turns or waits to the left of $c_1$ or to the right of $c_q$ in the original schedule, such a robot $R'$ must be moving in the same direction as $R$ towards $c_1$ (or $c_q$). Hence, all such robots can be placed in the at most $k-1$ vertices to the left of $c_1$ (right of $c_q$). If all $k$ robots move along the same row towards the image of $c_1$, we allow the first one to enter column $c_1$, as it does not block any other robot.

The robots then wait until the time step in which they enter $c_1$ (or $c_q$) in $\SSS$, at which point $R$ enters the image of $c_1$ (or $c_q$) in $\SSS'$. Whenever $R$ moves along some row between two columns $c_i$ and $c_{i+1}$ in $\SSS$ (for example, from $c_i$ to $c_{i+1}$, the reverse direction is symmetric), we again move $R$ as close as possible to the image of $c_{i+1}$ and let it wait there until it needs to enter $c_{i+1}$ in the original schedule, at which point we let $R$ enter the image of $c_{i+1}$. The motion of a robot along the image of a column $c \in \mathcal{C}$ is identical to that along $c$. 

Finally, when $R$ leaves $S'$, i.e., it moves left from $c_1$ (or right from $c_q$), we let it wait in the first $k$ columns to the left of $c_1$ (right of $c_q$) until the timestep at which it leaves the boundary of $S'$ in $\SSS$. At that point, we freeze the schedule, move $R$ to the boundary of $S'$ and let it continue according to $\SSS$. 

We now argue that the route of $R$ in $\SSS'$ does not incur additional cost compared to $\SSS$. Since the vertical movement/cost of $R$ is unchanged, it suffices to compare the horizontal movement traveled. Observe that $R$ must traverse $S'$ completely from one boundary to the other.  

Assume without loss of generality that $R$ is moving from the left boundary to the right boundary of $S'$. In $\SSS$, the total horizontal distance it travels is $\Delta+ 2\cdot (\sum_{j}\delta_j)$, where $\Delta$ is distance between the left and the right boundaries of $S'$ and $\delta_j$ is the length of the $j$-th segment in which $R$ moves from right to the left. Let $\delta_j'$ be the length of the corresponding $j$-th segment on which $R$ moves from the right to the left in the new schedule. By construction of the mapping, we have $\delta'_j\le \delta_j$.
Therefore, $R$ does not incur additional cost in $\SSS'$ compared to $\SSS$. 

We apply an analogous shifting procedure to all robots that enter or leave $S'$ in a direction $d \in {\uparrow,\downarrow}$, moving them into the bottommost $k\mu$ rows of $S'$. The same argument shows that these robots also do not increase their traveled distance. Moreover, at all times, at most one robot is in $S^*$. It remains to show that the robots moving in the left-most $k\cdot \mu$ columns of $S'$ do not collide with those moving in the bottom-most $k\cdot \mu$ rows of $S'$. Recall that the leftmost $k\cdot \mu$ columns are free; therefore, no robot that enters/leaves in direction in $\{\uparrow, \downarrow\}$ ever traverses these columns. Since the rerouting of these robots modifies only the rows they traverse and not the columns, no collision can occur. Analogously, the bottom-most $k\cdot\mu$ rows are free, and thus no robot that entering or leaving in a direction in $\{\leftarrow,\rightarrow\}$ ever occupies these rows. 
\end{proof}
\begin{figure}
    \centering
    \includegraphics[page=3]{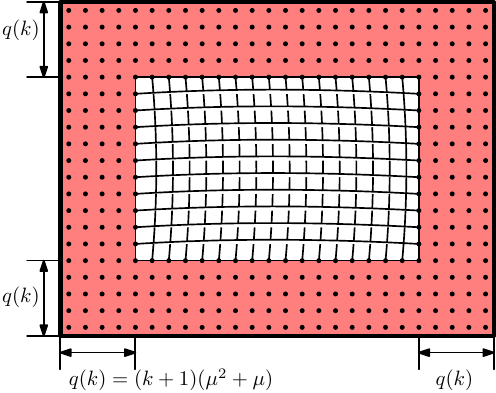}
    \caption{Illustration of Reduction Rule~\ref{red:rect}. In red is the part of the rectangle sector $S$ that is not inside the second frame $S^*$. The horizontal and vertical lines connecting the opposite sides of the second frame $S^*$ do not intersect each other. }
    \label{fig:red_rect}
\end{figure}
The following reduction rule follows from the above lemma; see Figure~\ref{fig:red_rect}.

\begin{reduction}
\label{red:rect}
Let $\I$ be an instance of \dcmp{}
and let $\Gamma$ be the sector graph of $\I$. 
 Let $\I'$ be the instance obtained by applying the following for every sector $S\in \Gamma$, where $S$ is an $a\times b$ rectangle sector with $\min\{a,b\}\geq (k+1)(\mu^2+\mu)$. 
Let $S^*$ be the second frame of $S$. First, remove all vertices of $S^*$, except the boundary vertices. Second, for every two boundary vertices of $S^*$ that are on opposite sides of $S^*$ and are on the same horizontal (resp.~vertical) line, join them by a path of degree-2 new vertices of the same length as their distance in $G$; see Figure~\ref{fig:red_rect}. Then $\I'$ is equivalent to $\I$.
\end{reduction}

We note that the above reduction rule does not produce a subgraph of the original graph; instead, it constructs a new graph that is neither planar nor smaller than the original.

\section{Bounding the Treewidth of the Graph}
\label{sec:treewidth}
We call a graph $G$ \emph{reduced} if none of Reduction Rules~\ref{red:hist}-\ref{red:rect} applies to $G$. 

\begin{figure}
    \centering
    \includegraphics[width=\linewidth,page=4]{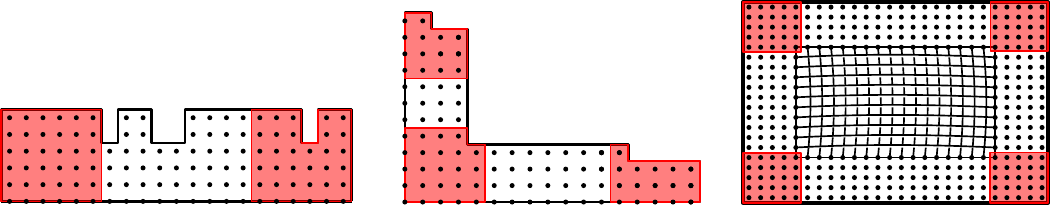}
    \caption{Illustration of initial special vertices for the three types of sector. For each sector $S$, initial special vertices split the sector into connected components such that each connected component has bounded treewidth and is connected either only to components in the horizontal neighbors of $S$ or only to components in the vertical neighbors of $S$. }
    \label{fig:initial_special_vertices}
\end{figure}
\begin{theorem}
\label{thm:treewidth}
Let $\I=(G, \R, k, \lambda)$ be an instance of \gcmp{} such that $G$ is reduced. Then $tw(G)$ is at most $7^{28k+5}\cdot q^3(k)$, where $q(k)= (k+1)(\mu^2+\mu)=k^{\Oh(k^2)}$.
\end{theorem}

\begin{proof}
Let $\Gamma$ be the sector graph of $\I$, and
let $S \in \Gamma$ be a non-clean sector.
We call a vertex $v$ an \emph{extension vertex} with respect to $S$ if there exists a $u$-$v$ path $P$ from $v$ to a vertex $u \in S$ such that all vertices in $P-u$ belong to clean sectors. We define the \emph{extended sector} $S^+$ of $S$ as the union of $S$ with the set of vertices that are extension vertices with respect to $S$. 

For a vertex $v \in S$, we say that vertex $u \in S^+\setminus S$ is an \emph{extension vertex} of $v$ if $u$ and $v$ lie on the same horizontal or vertical line of $G$. We define the \emph{extension} of set $X \subseteq S$, denoted $X^+$, to be the set of vertices in $S^+$ that are extensions of vertices in $X$. Note that, by Reduction Rule~\ref{red:cleansectors}, every extension vertex of $S$ is at a distance at most $\lambda/2$ from $S$.

Let $\Gamma'$ be the subgraph of $\Gamma$ consisting of the non-clean sectors. By definition of a clean sector, each connected component of $\Gamma[V(\Gamma)\setminus V(\Gamma')]$ is a tree with exactly one neighbor in $\Gamma'$.  By Lemma~\ref{lem:fewnonclean}, the maximum degree of $\Gamma'$ is at most $8$. Moreover, by Theorem~\ref{thm:sectortw}, the treewidth of $\Gamma'$ is at most $7^{\kappa-1}\le 7^{4k-1}$. 

By the result of \cite{GurskiW25}, if we take a tree decomposition  $(T, \chi)$ of $\Gamma'$ of width $\tw(\Gamma')$ and add to each bag $\chi(t)$, for $t\in V(T)$, all vertices in $\Gamma'$ at distance at most $3\kappa\le 12k$, we obtain a tree decomposition $(T, \chi')$ of $\Gamma'$ of width at most $(7^{4k-1}+1)\cdot (1+8\cdot \sum_{i=0}^{12k-1}7^i) -1\le  8\cdot7^{16k}$. Furthermore, by Lemma~\ref{lem:pathsectors}, any straight path in $G$ intersects at most $3\kappa$ sectors. Hence, for any straight path $P$ in $G$, there is a bag in $(T,\chi')$ that contains all non-clean sectors intersected by $P$. 

The tree decomposition $(T,\chi')$ serves as the starting point for constructing the tree decomposition of $G$. Our goal is to build a tree decomposition $(T', \chi'')$ satisfying the following properties:
\begin{itemize}
    \item $T$ is a subtree of $T'$;
    \item for $t\in T$, the bag $\chi''(t)$ contains $f(k)$ vertices from each extended sector in $\chi'(t)$, for some computable function $f$ specifies later; and 
    \item each tree in $T' - T$ corresponds to a straight path in $G$ and covers the connected component that contains that straight path. 
     \end{itemize}

Let $S\in \Gamma'$ be a sector and let $S^+$ be its extension. We define the \emph{initial special vertices} of $S$ to be all vertices in $S$ that lie in a square with side length $q(k)$ having an endpoint of a baseline of $S$ as corner (See Figure~\ref{fig:initial_special_vertices}). The \emph{initial special vertices} of $S^+$ are defined as the initial special vertices of $S$ together with their extensions. 

A vertex $v$ in $S^+$ is called \emph{special} if either: (1) it is an initial special vertex, or (2) it is at distance at most $q(k)$ from a baseline of $S$, and there exists a straight path through $v$ that avoids all initial special vertices in $S^+$ and contains an initial special vertex from an extension of a sector $S' \neq S$. 

Since $G$ is a subgrid and each sector has at most $4$ baseline endpoints, there are at most $4q^2(k)+4q(k)\lambda$ initial special vertices in each extended sector. 

Since any straight path intersects at most $12k$ sectors and each sector neighbors at most $8$ non-clean sectors, for every sector $S$, there are at most $8\cdot(7)^{12k}$ many sectors on any straight path through $S$. Every initial special vertex of a sector $S'$ on the same straight path as $S$ can be aligned with at most $2q(k)+\lambda$ vertices of $S^+$ that are within distance $q(k)$ from a baseline of $S$.

For each $t\in T$, we define $\chi''(t)$ to be the set of all special vertices in the extensions of sectors in $\chi'(t)$. We obtain the bound $|\chi''(t)|\le 8\cdot 7^{16k}\cdot (8\cdot 7^{12k} + 1)\cdot (2q(k)+\lambda)\cdot(4q^2(k)+4q(k)\lambda)\le 8^2\cdot 7^{28k+1}\cdot(8q^3(k)+12q^2(k)\cdot\lambda+4q(k)\lambda^2) \le  7^{28k+5}\cdot q^3(k)$. The previous upper bound was obtained using the fact that $\lambda \leq q(k)$.

We now consider the graph induced by the non-special vertices. Let $H$ be a connected component of this graph. We claim that $H$ is a subdivision of a subgrid of bounded height or width. Moreover, $H$ is contained in the extensions of at most $12k$ non-clean sectors intersecting a common straight path, and all special vertices adjacent to $H$ belong to these same sector extensions.

We first analyze how special vertices split each extended sector. If $S$ is a histogram or a staircase sector, then every vertex of $S$ lies at distance at most $q(k)$ from a baseline. By the definition of the baseline and the special vertices, it follows that if $P_S$ is a $d$-baseline for $d\in \directions$, then for every vertex $v\in P_S$, either all vertices (including $v$) that are reachable from $v$ in direction $d$ are special or none of them is special. Additionally, the first and last $q(k)$ vertices on $P_S$ are special. Consequently, removing the special vertices in $S^+$ splits it into connected components that contain a subpath of a single baseline, have bounded height (or width), and connect to other sectors only through that baseline. 

If $S$ is a rectangle sector, recall that it consists of a second frame containing all vertices within distance $q(k)$ from a boundary of $S$, together with horizontal and vertical degree-$2$ paths. Removing the initial special vertices splits the second frame into four disjoint subgrids of bounded height, top ($S_T$), bottom ($S_B$), left ($S_L$), and right ($S_R$). Each degree-2 path connects either$S_T$ to $S_B$ or $S_L$ to $S_R$. Moreover, if a special vertex lies in either $S_T$ or $S_B$, then all the at most $2q(k)$ second frame vertices on the same vertical line as $v$ are also special; an analogous statement holds for $S_L$ and $S_R$ along horizontal lines. 

Thus, the connected components obtained after removing special vertices from a rectangle sector can be classified as follows: 
\begin{itemize}
    \item degree-$2$ paths connecting two special vertices within $S$;
    
    \item \emph{horizontal} components that each contains a subgrid of $S_L$ and a subgrid of $S_R$, plus degree-2 horizontal paths connecting these two subgrids; and
    \item \emph{vertical} components that each contains a subgrid of $S_B$, a subgrid of $S_T$ and the degree-2 vertical paths connecting these two subgrids.
\end{itemize} 
Consider a vertical component $H$; the horizontal component case is symmetric. By definition of special vertices, no initial special vertex lies on the same vertical line as any vertex of $H$. Since the left-most and right-most vertices of any sector are initial special vertices, every vertical straight path through a vertex in $H$ intersects the same set of at most $12k$ sectors. Any connected components in other extended sectors intersected by these paths are either vertical components of a rectangle sector, or components of height at most $q(k)$ of a histogram or staircase sector. 

Therefore, the connected component containing $H$ consists of at most $24k$ subgrids, each of height at most $q(k)$, connected by vertical degree-two paths. Hence, this component is a subdivision of a subgrid of height at most $24k\cdot (q(k)+1)+\lambda$, where the $\lambda$ terms comes from the fact that the top-most and the bottom-most non-clean sector can be connected to clean sectors, that are contained in their extensions. Consequently, the connected component $C$ of $H$ has bounded treewidth, and its neighborhood in $G$ is contained in at most $12k$ extended sectors intersecting a common straight vertical path. Thus, there exists a bag in $(T,\chi'')$ containing $N(C)$, and we can attach a tree decomposition of $C\cup N(C)$ of width at most $24k\cdot (q(k)+1)+\lambda + |N(C)|\le 72k(q(k)+1)+\lambda$. 
 
 If a connected component does not contain part of any rectangle sector, the analysis is simpler, as the height or width of such components is already bounded by $(k+1)(\lambda+1)< q(k)$. Likewise, the tree decomposition for a connected component consisting of a degree-2 path connecting two special vertices in the same sector $S$ of width at most $3$ can be attached to the bag containing the special vertices of $S$.

 Finally, attaching all these tree decompositions to $(T,\chi'')$ does not increase the overall treewidth, since every attached bag has less than $7^{28k+5}\cdot q^3(k)$ vertices, which is the upper bound on the treewidth of $(T, \chi'')$. 
\end{proof}

\section{Putting it Together}
\label{sec:mainproof}
In this section, we combine the insights obtained above to prove Theorem~\ref{thm:main}. To this end, we recall that \cmpl\ is fixed-parameter tractable w.r.t.\ the treewidth of the input graph $G$ plus the number $k$ of robots~\cite[Theorem 2]{icalp}. 

\thmmain*
 
\begin{proof}
    Let $\I= (G, \R, k, \ell)$ be an instance of \cmpl{} on discretized polygons. We begin by applying Theorems~\ref{thm:fptapx}~and~\ref{thm:boundedturnsdistance} to obtain $g(k)$-many instances of \dcmp{} that satisfy Assumption~\ref{ass:canonical}. We solve each instance independently and return \YES~if and only if at least one of them is a \YES-instance. 
    
    In order to solve each one of these instances, we proceed as follows.
         We fix an arbitrary ordering of the ports, defined by the terminals of the instance, and compute the bend vectors of all vertices in $G$. This can be carried out easily in polynomial time by computing bend distances to each port via an iterative procedure:
    \begin{itemize}
        \item first, we compute all vertices reachable with $0$ bends; 
        \item then, for each $i\in [n]$, we compute the vertices reachable by $i+1$ bends by identifying all previously unreachable vertices that are reachable by a straight path from some vertex reachable with $i$ bends.
    \end{itemize}
     Using the bend vectors, we compute all sectors of $G$. This is achieved by partitioning the vertices according to their bend vectors. For each bend vector $\vec{v}$ that is realized by at least one vertex, we add to the set of sectors all connected components of the subgraph induced by the vertices with bend vector $\vec{v}$.
    We then exhaustively apply Reduction Rules~\ref{red:hist}-\ref{red:stair} to all sectors; this takes polynomial time. We follow that by the application of Reduction Rule~\ref{red:rect}.  
    Note that we do not recompute sectors after applying the reduction rules. 
    By Theorem~\ref{thm:treewidth}, after these reductions, the resulting graph has treewidth $k^{\Oh(k^2)}$, allowing us to apply Theorem 2 from~\cite{icalp} to the reduced instance to solve it. Since the algorithms in Theorem~\ref{thm:fptapx} and the aforementioned treewidth-based algorithm~\cite[Theorem 2]{icalp} run in \FPT{}-time, it follows that our algorithm runs in \FPT{}-time as well. 
\end{proof}

\fi
\section{Conclusion}
\label{sec:conclusion}
In this paper, we establish the fixed-parameter tractability of \textsc{Coordinated Motion Planning} on graphs arising from discretizations of simple polygons. Our result extends previously-known tractability results beyond full rectangular grids to a broader and more realistic class of graphs.

The central contribution of our work is a reduction of the input graph to one of bounded treewidth. We achieve this by introducing a novel notion of equivalence among vertices, defined in terms of having equivalent ``roadmaps'' to the set of designated terminal vertices. This notion enables us to partition the graph into regions, called sectors, analyze their structure, and derive appropriate reduction rules based on a characterization of these sectors. Notably, our reduction rules do not follow the standard paradigm of iteratively removing vertices. Instead, they go beyond vertex deletion by introducing auxiliary gadgets, resulting in a reduced graph that may fall outside the original graph class (e.g., violating planarity), while still having bounded treewidth. This contrasts with classical approaches, which typically obtain bounded treewidth by successively deleting vertices from the original graph.

It is not too difficult to extend the \FPT{} result presented here to full grids containing a single hole, which can be viewed as a “flipped” discretized polygon. This naturally leads to the following open questions:
\begin{itemize}
\item[(i)] Can the \FPT{} techniques developed in this paper be extended to full grids with a constant number of holes, or even a parameter-bounded number of holes?
\item[(ii)] Is the \gcmp{} problem \FPT{} on subgrids, or even on planar graphs?
 \end{itemize}

We strongly believe that the answer to the first question is affirmative and that a fixed-parameter algorithm is attainable by refining the approach introduced here. For question~(ii), we conjecture that the answer remains affirmative as well, and that reducing the graph to one of bounded treewidth remains the correct high-level strategy. Such a result, however, would likely require new insights, including more sophisticated vertex classifications and novel treewidth reduction rules. In our view, this constitutes the most pressing open problem arising from the study of the \gcmp{} problem. Finally, we remark that fixed-parameter (in)tractability for \gcmp\ w.r.t.\ the number of agents is not settled even for general graphs.

\bibliography{ref}

@ARTICLE{grid2,
  author={Solis, Irving and Motes, James and Sandström, Read and Amato, Nancy M.},
  journal={IEEE Robotics and Automation Letters}, 
  title={Representation-Optimal Multi-Robot Motion Planning Using Conflict-Based Search}, 
  year={2021},
  volume={6},
  number={3},
  pages={4608-4615}}

@article{grid1,
title = {An empirical evaluation of learning-based multi-agent path finding algorithms in warehouse environments},
author={Rivas, Alberto and Alshamsi, Mohamed and Garc{\'\i}a, Sergio and Pelliccione, Patrizio and Chicano, Francisco},
journal = {Robotics and Autonomous Systems},
volume = {194},
pages = {105149},
year = {2025}
}

@inproceedings{0001TKDKK21,
  author       = {Jiaoyang Li and
                  Andrew Tinka and
                  Scott Kiesel and
                  Joseph W. Durham and
                  T. K. Satish Kumar and
                  Sven Koenig},
  title        = {Lifelong Multi-Agent Path Finding in Large-Scale Warehouses},
  booktitle    = {Thirty-Fifth {AAAI} Conference on Artificial Intelligence, {AAAI}
                  2021, Thirty-Third Conference on Innovative Applications of Artificial
                  Intelligence, {IAAI} 2021, The Eleventh Symposium on Educational Advances
                  in Artificial Intelligence, {EAAI} 2021, Virtual Event, February 2-9,
                  2021},
  pages        = {11272--11281},
  publisher    = {{AAAI} Press},
  year         = {2021},
  url          = {https://doi.org/10.1609/aaai.v35i13.17344},
  doi          = {10.1609/AAAI.V35I13.17344},
  bibsource    = {dblp computer science bibliography, https://dblp.org}
}

@inproceedings{MorrisPLMMKK16,
  author       = {Robert Morris and
                  Corina S. Pasareanu and
                  Kasper S{\o}e Luckow and
                  Waqar Malik and
                  Hang Ma and
                  T. K. Satish Kumar and
                  Sven Koenig},
  editor       = {Daniele Magazzeni and
                  Scott Sanner and
                  Sylvie Thi{\'{e}}baux},
  title        = {Planning, Scheduling and Monitoring for Airport Surface Operations},
  booktitle    = {Planning for Hybrid Systems, Papers from the 2016 {AAAI} Workshop,
                  Phoenix, Arizona, USA, February 13, 2016},
  series       = {{AAAI} Technical Report},
  volume       = {{WS-16-12}},
  publisher    = {{AAAI} Press},
  year         = {2016},
  url          = {http://www.aaai.org/ocs/index.php/WS/AAAIW16/paper/view/12611},
  bibsource    = {dblp computer science bibliography, https://dblp.org}
}

@inproceedings{VelosoBCR15,
  author       = {Manuela M. Veloso and
                  Joydeep Biswas and
                  Brian Coltin and
                  Stephanie Rosenthal},
  editor       = {Qiang Yang and
                  Michael J. Wooldridge},
  title        = {CoBots: Robust Symbiotic Autonomous Mobile Service Robots},
  booktitle    = {Proceedings of the Twenty-Fourth International Joint Conference on
                  Artificial Intelligence, {IJCAI} 2015, Buenos Aires, Argentina, July
                  25-31, 2015},
  pages        = {4423},
  publisher    = {{AAAI} Press},
  year         = {2015},
  url          = {http://ijcai.org/Abstract/15/656},
  bibsource    = {dblp computer science bibliography, https://dblp.org}
}

@inproceedings{FKMMO25,
  title={Exact algorithms for multiagent path finding with communication constraints on tree-like structures},
  author={Fioravantes, Foivos and Knop, Du{\v{s}}an and K{\v{r}}i{\v{s}}t'an, Jan Maty{\'a}{\v{s}} and Melissinos, Nikolaos and Opler, Michal},
  booktitle={Proceedings of the AAAI Conference on Artificial Intelligence},
  volume={39},
  number={22},
  pages={23177--23185},
  year={2025}
}

@inproceedings{FioravantesKKMO24,
  author       = {Foivos Fioravantes and
                  Dušan Knop and
                  Jan Matyáš Křišťan and
                  Nikolaos Melissinos and
                  Michal Opler},
  editor       = {Michael J. Wooldridge and
                  Jennifer G. Dy and
                  Sriraam Natarajan},
  title        = {Exact Algorithms and Lowerbounds for Multiagent Path Finding: Power
                  of Treelike Topology},
  booktitle    = {Thirty-Eighth {AAAI} Conference on Artificial Intelligence},
  pages        = {17380--17388},
  publisher    = {{AAAI} Press},
  year         = {2024},
  url          = {https://doi.org/10.1609/aaai.v38i16.29686} 
}

@inproceedings{EGKS25SoCG,
  author       = {Eduard Eiben and
                  Robert Ganian and
                  Iyad Kanj and
                  M. S. Ramanujan},
  editor       = {Oswin Aichholzer and
                  Haitao Wang},
  title        = {A Minor-Testing Approach for Coordinated Motion Planning with Sliding
                  Robots},
  booktitle    = {41st International Symposium on Computational Geometry, SoCG 2025,
                  Kanazawa, Japan, June 23-27, 2025},
  series       = {LIPIcs},
  volume       = {332},
  pages        = {44:1--44:15},
  publisher    = {Schloss Dagstuhl - Leibniz-Zentrum f{\"{u}}r Informatik},
  year         = {2025},
  url          = {https://doi.org/10.4230/LIPIcs.SoCG.2025.44},
  doi          = {10.4230/LIPICS.SOCG.2025.44},
  bibsource    = {dblp computer science bibliography, https://dblp.org}
}

@inproceedings{icalp,
  author       = {Argyrios Deligkas and
                  Eduard Eiben and
                  Robert Ganian and
                  Iyad Kanj and
                  M. S. Ramanujan},
  editor       = {Karl Bringmann and
                  Martin Grohe and
                  Gabriele Puppis and
                  Ola Svensson},
  title        = {Parameterized Algorithms for Coordinated Motion Planning: Minimizing
                  Energy},
  booktitle    = {51st International Colloquium on Automata, Languages, and Programming,
                  {ICALP} 2024, Tallinn, Estonia, July 8-12, 2024},
  series       = {LIPIcs},
  volume       = {297},
  pages        = {53:1--53:18},
  publisher    = {Schloss Dagstuhl - Leibniz-Zentrum f{\"{u}}r Informatik},
  year         = {2024},
note = {Full version available at \url{https://arxiv.org/abs/2404.15950}}
}

@article{kiva,
author = {Wurman, Peter and D'Andrea, Raffaello and Mountz, Mick},
year = {2008},
month = {03},
pages = {9-20},
title = {Coordinating Hundreds of Cooperative, Autonomous Vehicles in Warehouses.},
volume = {29},
journal = {AI Magazine}
}

@article{sharir1,
  title={{On the “piano movers'” problem I. The case of a two-dimensional rigid polygonal body moving amidst polygonal barriers}},
  author={Schwartz, Jacob T. and Sharir, Micha},
  journal={Communications on Pure and Applied Mathematics},
  volume={36},
  number={3},
  pages={345--398},
  year={1983},
  publisher={Wiley Online Library}
}

@article{sharir2,
  title={{On the “piano movers” problem. II. General techniques for computing topological properties of real algebraic manifolds}},
  author={Schwartz, Jacob T. and Sharir, Micha},
  journal={Advances in Applied Mathematics},
  volume={4},
  number={3},
  pages={298--351},
  year={1983},
  publisher={Elsevier}
}

@INPROCEEDINGS{alagar,
  author={Ramanathan, Geetha and Alagar, Vangalur S.},
  booktitle={ICRA}, 
  title={Algorithmic motion planning in robotics: Coordinated motion of several disks amidst polygonal obstacles}, 
  year={1985},
  volume={2},
  number={},
  pages={514-522} }

@inproceedings{GeftHalperin,
author = {Geft, Tzvika and Halperin, Dan},
title = {Refined Hardness of Distance-Optimal Multi-Agent Path Finding},
year = {2022},
isbn = {9781450392136},
booktitle = {AAMAS},
pages = {481–488},
numpages = {8},
location = {Virtual Event, New Zealand}
}

@book{DowneyFellows13,
  author    = {Rodney G. Downey and
               Michael R. Fellows},
  title     = {Fundamentals of Parameterized Complexity},
  publisher = {Springer},
  series    = {Texts in Computer Science},
  year      = {2013} }

@book{CyganFKLMPPS15,
  author    = {Marek Cygan and
               Fedor V. Fomin and
               Lukasz Kowalik and
               Daniel Lokshtanov and
               D{\'{a}}niel Marx and
               Marcin Pilipczuk and
               Michal Pilipczuk and
               Saket Saurabh},
  title     = {Parameterized Algorithms},
  publisher = {Springer},
  year      = {2015}
}

@BOOK{FlumGrohe06,
  title = {Parameterized Complexity Theory},
  publisher = {Springer},
  year = {2006},
  author = {J\"{o}rg Flum and Martin Grohe},
  volume = {XIV},
  series = {Texts in Theoretical Computer Science. An EATCS Series},
  address = {Berlin}
}

@article{heuristic1,
title = {Conflict-based search for optimal multi-agent pathfinding},
journal = {Artificial Intelligence},
volume = {219},
pages = {40-66},
year = {2015},
 author = {Guni Sharon and Roni Stern and Ariel Felner and Nathan R. Sturtevant}}

@ARTICLE{heuristic2,
  author={Hart, Peter E. and Nilsson, Nils J. and Raphael, Bertram},
  journal={IEEE Transactions on Systems Science and Cybernetics}, 
  title={A Formal Basis for the Heuristic Determination of Minimum Cost Paths}, 
  year={1968},
  volume={4},
  number={2},
  pages={100-107}}

@inproceedings{heuristic3,
  author={Eli Boyarski and Ariel Felner and Roni Stern and Guni Sharon and David Tolpin and Oded Betzalel and Solomon Eyal Shimony},
  title={{ICBS}: Improved Conflict-Based Search Algorithm for Multi-Agent Pathfinding},
  year={2015},
  pages={740-746}, 
  booktitle={IJCAI}
}

@article{heuristic4,
title = {Subdimensional expansion for multirobot path planning},
journal = {Artificial Intelligence},
volume = {219},
pages = {1-24},
year = {2015}, 
author = {Glenn Wagner and Howie Choset}}

@article{sharir,
  title={On the Piano Movers' Problem: {III.} Coordinating the Motion of Several Independent Bodies: The Special Case of Circular Bodies Moving Amidst Polygonal Barriers},
  author={Jacob T. Schwartz and Micha Sharir},
  journal={The International Journal of Robotics Research},
  year={1983},
  volume={2},
  pages={46 - 75}
}

@ARTICLE{lavalle1,
  author={Yu, Jingjin and LaValle, Steven M.},
  journal={IEEE Transactions on Robotics}, 
  title={Optimal Multirobot Path Planning on Graphs: Complete Algorithms and Effective Heuristics}, 
  year={2016},
  volume={32},
  number={5},
  pages={1163-1177}
  }

@article{nphard1,
author = {Daniel Ratner and Manfred Warmuth},
title = {The ($n^2-1$)-puzzle and related relocation problems},
journal = {Journal of Symbolic Computation},
volume = {10},
number = {2},
pages = {111-137},
year = {1990} 
}

@article{nphard2,
author={Erik D. Demaine and Mikhail Rudoy},
  title={A simple proof that the ($n^2-1$)-puzzle is hard},  
  journal={Theoretical Computer Science},
  year={2018},
  volume={732},
  pages={80-84}
}

@inproceedings{lavalle,
  author    = {Jingjin Yu and
               Steven M. LaValle},
  title     = {Structure and Intractability of Optimal Multi-Robot Path Planning
               on Graphs},
  booktitle = {{AAAI}},
  pages={1443--1449},
  year      = {2013}
  
}

@article{demaine,
  author    = {Erik D. Demaine and
               S{\'{a}}ndor P. Fekete and
               Phillip Keldenich and
               Henk Meijer and
               Christian Scheffer},
  title     = {Coordinated Motion Planning: Reconfiguring a Swarm of Labeled Robots with Bounded Stretch},
  journal   = {{SIAM} Journal on Computing},
  volume    = {48},
  number    = {6},
  pages     = {1727--1762},
  year      = {2019}
}

@ARTICLE{banfi,
  author={Banfi, Jacopo and Basilico, Nicola and Amigoni, Francesco},
  journal={IEEE Robotics and Automation Letters}, 
  title={Intractability of Time-Optimal Multirobot Path Planning on {2D} Grid Graphs with Holes}, 
  year={2017},
  volume={2},
  number={4},
  pages={1941-1947}}

@inproceedings{halprinunlabeled,
  author    = {Bahareh Banyassady and
               Mark de Berg and
               Karl Bringmann and
               Kevin Buchin and
               Henning Fernau and
               Dan Halperin and
               Irina Kostitsyna and
               Yoshio Okamoto and
               Stijn Slot},
  
  title     = {Unlabeled Multi-Robot Motion Planning with Tighter Separation Bounds},
  booktitle = {{SoCG}},
  volume    = {224},
  pages     = {12:1--12:16},
  year      = {2022} 
}

@incollection{survey,
author = {Adrian Dumitrescu},
title = {Motion Planning and Reconfiguration for Systems of Multiple Objects},
booktitle = {Mobile Robots},
publisher = {IntechOpen},
address = {Rijeka},
year = {2007},
editor = {Sascha Kolski},
chapter = {24}
}

@book{Diestel12,
  author    = {Reinhard Diestel},
  title     = {Graph Theory, 4th Edition},
  series    = {Graduate texts in mathematics},
  volume    = {173},
  publisher = {Springer},
  year      = {2012}
}

@inproceedings{EibenGK23,
  author       = {Eduard Eiben and
                  Robert Ganian and
                  Iyad Kanj},
  title        = {The Parameterized Complexity of Coordinated Motion Planning},
  booktitle    = {SoCG},
  volume       = {258},
  pages        = {28:1--28:16},
  year         = {2023},
note = {full version at https://arxiv.org/abs/2312.07144}
}

@inproceedings{PapadimitriouRST94,
  author       = {Christos H. Papadimitriou and
                  Prabhakar Raghavan and
                  Madhu Sudan and
                  Hisao Tamaki},
  title        = {Motion Planning on a Graph (Extended Abstract)},
  booktitle    = {STOC},
  pages        = {511--520},
  year         = {1994}
}

@article{socg2021,
  author       = {S{\'{a}}ndor P. Fekete and
                  Phillip Keldenich and
                  Dominik Krupke and
                  Joseph S. B. Mitchell},
  title        = {Computing Coordinated Motion Plans for Robot Swarms: The {CG:} {SHOP}
                  Challenge 2021},
  journal      = {{ACM} Journal on Experimental Algorithmics},
  volume       = {27},
  pages        = {3.1:1--3.1:12},
  year         = {2022}
}

@inproceedings{BhoreGKMN23,
  author       = {Sujoy Bhore and
                  Robert Ganian and
                  Liana Khazaliya and
                  Fabrizio Montecchiani and
                  Martin N{\"{o}}llenburg},
  editor       = {Erin W. Chambers and
                  Joachim Gudmundsson},
  title        = {Extending Orthogonal Planar Graph Drawings Is Fixed-Parameter Tractable},
  booktitle    = {39th International Symposium on Computational Geometry, SoCG 2023,
                  June 12-15, 2023, Dallas, Texas, {USA}},
  series       = {LIPIcs},
  volume       = {258},
  pages        = {18:1--18:16},
  publisher    = {Schloss Dagstuhl - Leibniz-Zentrum f{\"{u}}r Informatik},
  year         = {2023},
  url          = {https://doi.org/10.4230/LIPIcs.SoCG.2023.18},
  doi          = {10.4230/LIPICS.SOCG.2023.18},
  bibsource    = {dblp computer science bibliography, https://dblp.org}
}

@article{DeligkasEGKLR26,
  author       = {Argyrios Deligkas and
                  Eduard Eiben and
                  Robert Ganian and
                  Iyad Kanj and
                  Dominik Leko and
                  M. S. Ramanujan},
  title        = {Routing few robots in a crowded network},
  journal      = {J. Comput. Syst. Sci.},
  volume       = {157},
  pages        = {103753},
  year         = {2026},
  url          = {https://doi.org/10.1016/j.jcss.2025.103753},
  doi          = {10.1016/J.JCSS.2025.103753},
  bibsource    = {dblp computer science bibliography, https://dblp.org}
}

@inproceedings{SODA26new,
  author       = {Benjamin Holmgren and Pankaj K. Agarwal and Alex Steiger},
  title        = {Near-Optimal Min-Sum Multi-Robot Motion Planning in a Planar Polygonal Environment},
  booktitle    = {Proceedings of the 2026 Annual {ACM-SIAM} Symposium on Discrete Algorithms,
                  {SODA} 2026},
  publisher    = {{SIAM}},
  year         = {2026},
note = {to appear}
}

@article{AdlerKKLST17,
  author       = {Isolde Adler and
                  Stavros G. Kolliopoulos and
                  Philipp Klaus Krause and
                  Daniel Lokshtanov and
                  Saket Saurabh and
                  Dimitrios M. Thilikos},
  title        = {Irrelevant vertices for the planar Disjoint Paths Problem},
  journal      = {J. Comb. Theory {B}},
  volume       = {122},
  pages        = {815--843},
  year         = {2017},
  url          = {https://doi.org/10.1016/j.jctb.2016.10.001},
  doi          = {10.1016/J.JCTB.2016.10.001},
  bibsource    = {dblp computer science bibliography, https://dblp.org}
}

@inproceedings{SODA26paths,
  author       = {Michał Pilipczuk and Giannos Stamoulis and Michał Włodarczyk},
  title        = {Planar Disjoint Shortest Paths is Fixed-Parameter Tractable},
  booktitle    = {Proceedings of the 2026 Annual {ACM-SIAM} Symposium on Discrete Algorithms,
                  {SODA} 2026},
  publisher    = {{SIAM}},
  year         = {2026},
note = {to appear}
}

@article{GurskiW25,
  author       = {Frank Gurski and
                  Robin Weishaupt},
  title        = {The Behavior of Tree-Width and Path-Width Under Graph Operations and
                  Graph Transformations},
  journal      = {Algorithms},
  volume       = {18},
  number       = {7},
  pages        = {386},
  year         = {2025},
  url          = {https://doi.org/10.3390/a18070386},
  doi          = {10.3390/A18070386},
}

@inproceedings{DeligkasEGK025,
  author       = {Argyrios Deligkas and
                  Eduard Eiben and
                  Robert Ganian and
                  Iyad Kanj and
                  M. S. Ramanujan},
  editor       = {Sanmay Das and
                  Ann Now{\'{e}} and
                  Yevgeniy Vorobeychik},
  title        = {Parameterized Algorithms for Multiagent Pathfinding on Trees},
  booktitle    = {Proceedings of the 24th International Conference on Autonomous Agents
                  and Multiagent Systems, {AAMAS} 2025, Detroit, MI, USA, May 19-23,
                  2025},
  pages        = {584--592},
  publisher    = {International Foundation for Autonomous Agents and Multiagent Systems
                  / {ACM}},
  year         = {2025},
  url          = {https://dl.acm.org/doi/10.5555/3709347.3743574},
  doi          = {10.5555/3709347.3743574},
  bibsource    = {dblp computer science bibliography, https://dblp.org}
}

@article{DemaineHM14,
  author       = {Erik D. Demaine and
                  Mohammad Taghi Hajiaghayi and
                  D{\'{a}}niel Marx},
  title        = {Minimizing Movement: Fixed-Parameter Tractability},
  journal      = {{ACM} Trans. Algorithms},
  volume       = {11},
  number       = {2},
  pages        = {14:1--14:29},
  year         = {2014},
  url          = {https://doi.org/10.1145/2650247},
  doi          = {10.1145/2650247},
}
 \end{document}